\documentclass[sigconf, nonacm]{acmart}

\setcopyright{none}
\AtBeginDocument{%
  \providecommand\BibTeX{{%
    \normalfont B\kern-0.5em{\scshape i\kern-0.25em b}\kern-0.8em\TeX}}}

\usepackage[scaled=.8]{beramono}
\usepackage{booktabs}
\usepackage{color}
\usepackage{colortbl}
\usepackage{tabularray}
\usepackage{float}
\AtBeginDocument{%
  \providecommand\BibTeX{{%
    \normalfont B\kern-0.5em{\scshape i\kern-0.25em b}\kern-0.8em\TeX}}}

\usepackage{makecell}
\usepackage{soul}
\usepackage{bbm}
\usepackage{amsthm}
\usepackage{balance}
\usepackage{multirow}
\usepackage[ruled,linesnumbered,vlined]{algorithm2e}
\usepackage{tcolorbox}
\usepackage[noend]{algpseudocode}
\usepackage[flushleft]{threeparttable}
\usepackage{fbox}
\usepackage{xcolor}
\usepackage{wasysym}
\usepackage{booktabs}
\usepackage{verbatim}
\usepackage{pgfplots}
\pgfplotsset{compat=1.18}
\usepackage{hyperref}
\usepackage{fontawesome}
\usepackage{tikz}
\usepackage{graphics}
\usepackage{multirow}
\usepackage{fancybox}
\usepackage{listings}
\usepackage{enumitem}
\usetikzlibrary{patterns}
\usepackage{amsmath, amsfonts}
\usepackage{graphicx}
\usepackage{subfigure}
\usepackage{textcomp}
\usepackage{xcolor}
\usepackage{cleveref}

\newtheorem{problem}{Problem}
\newtheorem{mydef}{Definition}[section]

\newcommand{\defeq}{\overset{\text{\tiny def}}{=}}
\newcommand{\GSI}{\texttt{GSI}}
\renewcommand\footnotetextcopyrightpermission[1]{}
\definecolor{add}{rgb}{0.2,0.4,0.6}

\begin{document}

\title{Efficient Historical Butterfly Counting in Large Temporal Bipartite Networks via Graph Structure-aware Index}




\textfloatsep 1mm plus 1mm \intextsep 1mm plus 1mm

\author{Qiuyang Mang}
\authornotemark[1]
\affiliation{%
  \institution{CUHK-Shenzhen}
\country{}
}
\email{qiuyangmang@link.cuhk.edu.cn}

\author{Jingbang Chen}
\authornotemark[1]
\affiliation{%
\country{}
\institution{University of Waterloo}
}
\email{j293chen@uwaterloo.ca}

\author{Hangrui Zhou}
\authornote{The first three authors contributed equally to this research.}
\affiliation{%
\country{}
\institution{Tsinghua University}
}
\email{zhouhr23@mails.tsinghua.edu.cn}

\author{Yu Gao}
\affiliation{%
  \institution{Independent}
\country{}
}
\email{ygao2606@gmail.com}

\author{Yingli Zhou}
\affiliation{%
  \institution{CUHK-Shenzhen}
\country{}
}
\email{yinglizhou@link.cuhk.edu.cn}

\author{Qingyu Shi}
\affiliation{%
  \institution{Independent}
\country{}
}
\email{qingyuqwq@gmail.com}

\author{Richard Peng}
\affiliation{%
  \institution{Carnegie Mellon University}
\country{}
}
\email{yangp@cs.cmu.edu}

\author{Yixiang Fang}
\affiliation{%
  \institution{CUHK-Shenzhen}
\country{}
}
\email{fangyixiang@cuhk.edu.cn}

\author{Chenhao Ma}
\affiliation{%
  \institution{CUHK-Shenzhen}
\country{}
}
\email{machenhao@cuhk.edu.cn}




\captionsetup{font=small}
\setlength{\abovecaptionskip}{2.5pt}   
\setlength{\belowcaptionskip}{3pt}   
\setlength{\textfloatsep}{2pt}
\setlength{\floatsep}{2pt}
\everymath{\small}       
\everydisplay{\small}   

\newcommand{\todo}[1]{\textcolor{red}{[todo: #1]}}
\newcommand{\jb}[1]{\textcolor{blue}{[cjb: #1]}}
\newcommand{\yu}[1]{\textcolor{green}{[gy: #1]}}
\newcommand{\mqy}[1]{\textcolor{violet}{[mqy: #1]}}
\newcommand{\mqytext}[1]{\textcolor{violet}{#1}}
\newcommand{\hhz}[1]{\textcolor{teal}{[hhz: #1]}}
\newcommand{\mch}[1]{\textcolor{purple}{[mch: #1]}}
\newcommand{\zhou}[1]{\textcolor{magenta}{[zhou: #1]}}
\newcommand{\pr}[2]{\mathbb{P}_{#1}\left[#2\right]}
\newcommand{\expec}[2]{\mathbb{E}_{#1}\left[#2\right]}
\newcommand{\Otil}{\widetilde{O}}
\newcommand{\modify}[1]{{#1}}
\begin{abstract}
Bipartite graphs are ubiquitous in many domains, e.g., e-commerce platforms, social networks, and academia, by modeling interactions between distinct entity sets. Within these graphs, the butterfly motif, a complete 2$\times$2 biclique, represents the simplest yet significant subgraph structure, crucial for analyzing complex network patterns. Counting the butterflies offers significant benefits across various applications, including community analysis and recommender systems. 
Additionally, the temporal dimension of bipartite graphs, where edges activate within specific time frames, introduces the concept of historical butterfly counting, i.e., counting butterflies within a given time interval. This temporal analysis sheds light on the dynamics and evolution of network interactions, offering new insights into their mechanisms. 
Despite its importance, no existing algorithm can efficiently solve the historical butterfly counting task. To address this, we design two novel indices whose memory footprints are dependent on \#butterflies and \#wedges, respectively. Combining these indices, we propose a graph structure-aware indexing approach that significantly reduces memory usage while preserving exceptional query speed.
To further reduce the index size and boost the query efficiency, we design an index compression strategy, enabling the fast, high-quality, and unbiased approximation of historical butterfly counts.
We theoretically prove that our approach is particularly advantageous on power-law graphs, a common characteristic of real-world bipartite graphs, by surpassing traditional complexity barriers for general graphs.
Extensive experiments reveal that our query algorithms outperform existing methods by up to five magnitudes, effectively balancing speed with manageable memory requirements.
\end{abstract}

\maketitle

\section{INTRODUCTION}


Due to its ability to model relationships between two distinct sets of entities, the bipartite graph, or network, holds significant importance across various fields, including disease control on people-location networks \cite{eubank2004modelling,pavlopoulos2018bipartite}, fraud detection on user-page networks \cite{wang2019vertex,liu2019efficient} and recommendation on customer-product networks \cite{wang2006unifying,he2017neural,he2020lightgcn,wang2020efficient,wu2022graph}.
To analyze the structure and dynamics of the network, counting motifs is one of the most popular methods \cite{milo2002network,wang2020efficient,fang2020survey,jha2013space,eubank2004modellingm} since motifs are considered the basic construction block of the network. The \textit{butterfly} motif ($2 \times 2$ biclique) represents fundamental interaction patterns within the graph. Counting it has wide applications ranging from biological ecosystems \cite{wang2019vertex,sanei2018butterfly,wang2014rectangle}, where it helps in identifying mutual relationships between species, to social networks \cite{zou2016bitruss,wang2020efficient}, where it uncovers patterns of collaborations and affiliations.

Temporal bipartite graphs, in which edges typically carry timestamps, are often considered \cite{cai2024efficient, chen2021efficiently, eubank2004modellingm, pavlopoulos2018bipartite} since real-world interactions (modeled as edges) usually occur at specific timestamps.
Recently, Cai et al. \cite{cai2024efficient} first considered counting butterflies on temporal bipartite graphs, which extends the analytical depth of traditional bipartite graph analyses by incorporating the dimension of time, making it a powerful tool for uncovering dynamic patterns in complex systems.

\begin{figure}[t!]
    \centering
    \includegraphics[width=0.8\linewidth]{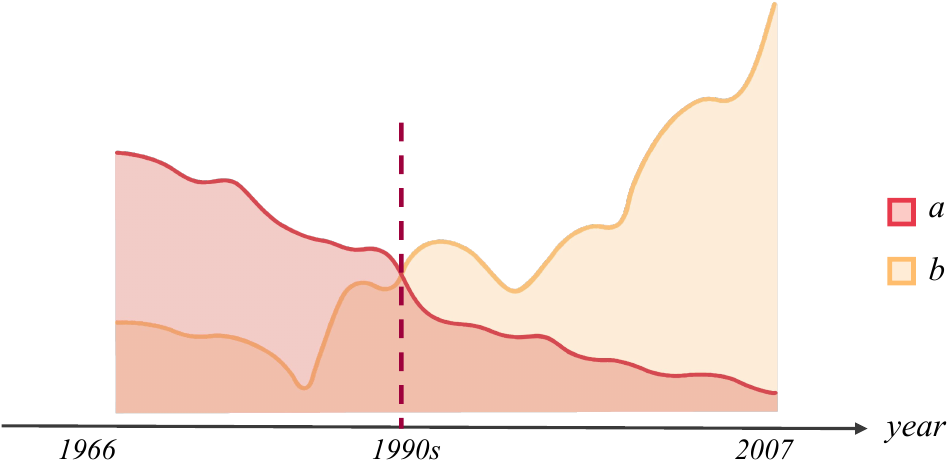}
    \caption{Jim Gray's activeness in database community (a) and astronomy (b) community.}
    \label{fig:intro_plot}
\end{figure}

However, merely counting butterflies across the entire timeline, as suggested by \cite{cai2024efficient}, may not accurately reflect the dynamic nature of relationships, failing to capture evolving trends. To address this, it is essential to consider the temporal dimension of interactions. By focusing on the occurrence of motifs within specified time frames, we introduce the concept of \textit{historical butterfly counting}. This approach, which involves analyzing butterflies within discrete time windows on a temporal bipartite graph, offers enhanced insights into the timing and progression of interactions.
It provides an in-depth understanding of network dynamics, uncovering the mechanisms behind network evolution and revealing opportunities for precise interventions across various domains.

\modify{Consider the example shown in Figure \ref{fig:intro_plot}. By applying historical butterfly counting, we capture the emerging and diminishing trends in Jim Gray’s research interests and his level of engagement within the database and astronomy communities. If we were to ignore the time dimension and simply count butterflies over the entire period from 1966 to 2007, we would obtain only two cumulative butterfly counts for each community. This aggregate approach would indicate Jim’s activity across both fields but fail to reveal his gradual shift in focus over time.}
\modify{In contrast, historical butterfly counting provides a more detailed view, allowing us to pinpoint specific periods of change, such as the gradual shift from databases to astronomy around the 1990s.}
\modify{Hence, the historical butterfly counts offer a more granular view of evolving communities and reveal periods of gradual or sudden changes.}

\modify{Regarding the approximate solution, we brief that in many applications of butterfly counting—such as graph kernel analysis~\cite{sheshbolouki2022sgrapp} and network measurement tasks like detecting dense substructures in recommendation systems~\cite{sanei2018butterfly, wang2014rectangle}—an approximate count is often sufficient. Here, we expand some discussions:}

\modify{In {\bf recommendation systems}, butterfly counting can reveal dense substructures among users or items, such as groups of users with overlapping interests or items frequently co-purchased. For example, approximate butterfly counts can be used to identify clusters of products that are likely to be bought together, enhancing cross-selling recommendations without needing perfectly exact counts. Given the massive scale of data in e-commerce or streaming platforms, allowing a small error margin can lead to substantial memory savings, enabling real-time processing and improved system scalability.}

\modify{Approximate butterfly counts are also valuable in {\bf graph kernel methods}, where motifs are used to create similarity measures between graphs for tasks such as classification or anomaly detection. Since graph kernels rely on relative similarities rather than exact counts, approximate butterfly counts are sufficient to maintain kernel performance while reducing memory usage and computational cost. This approach enables efficient similarity computations in fields like bioinformatics (e.g., comparing protein interaction networks) or cybersecurity (e.g., identifying similar attack patterns across networks).}

\modify{\hspace{-1em}\textbf{\textit{Applications}}
We now discuss some other interesting applications of historical butterfly counting:}

$\bullet$ \textbf{Bipartite Clustering Coefficient (BCC) Computation.} The \emph{bipartite clustering coefficient}~\cite{lind2005cycles,aksoy2017measuring, opsahl2013triadic} is a traditional cohesiveness measure for bipartite graphs whose computing bottleneck is counting butterflies.
Specifically, considering the scientific collaboration network (modeled as a temporal graph) in the given time window, a higher BCC suggests a stronger trend of cohesive collaboration within the research community. This coefficient positively correlates with the number of author pairs publishing multiple publications within a given time-window and inversely correlates with the number of author pairs collaborating only once.
For example, in academia, a scientist may change their frequent collaborators or research interests over time. 
We also provide a case study on the global research collaboration trend and close collaboration time windows in \cref{sec:casestudy} with BCC computation.

$\bullet$ \textbf{Identifying Close Communication Time-windows.}
Directly counting butterflies can identify time windows during close communication within a specific community.
For instance, we are interested in close collaboration.
Among Kaiming He's 2-hop ego networks w.r.t. all possible time-windows, we find that the two-year time-window with the highest butterfly count is 2010 - 2011, with the corresponding research records detailed in \Cref{fig:case-study-2}. During this period, Kaiming He, Jian Sun, Xiaoou Tang, Carsten Rother, and Christoph Rhemann established a series of cohesive research collaborations, resulting in the publication of three papers. In other words, \Cref{fig:case-study-2} depicts Kaiming's 2-hop ego network during 2010 - 2011, the two-year time interval with the highest butterfly count. Notably, 2010 and 2011, the last two years of Kaiming's Ph.D. under supervisor Xiaoou, marked his most cohesive collaborations.

\begin{figure}[t!]
    \centering
    \includegraphics[width=0.8\linewidth]{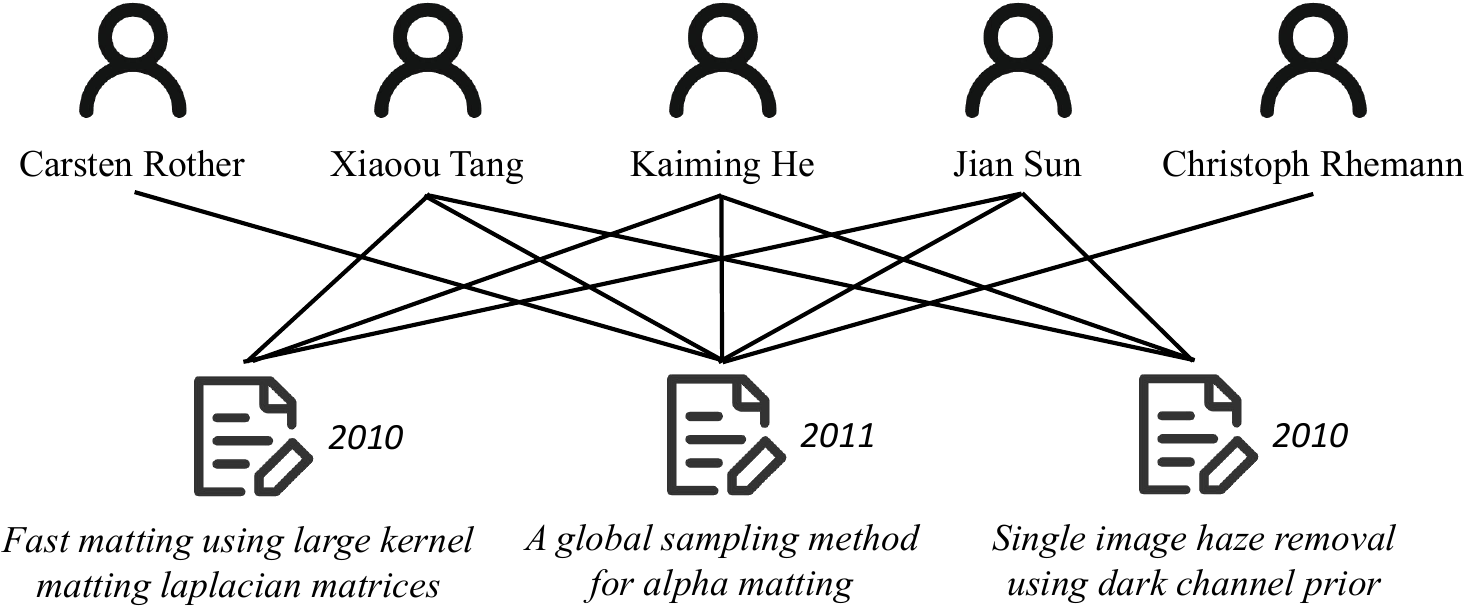}
    \caption{Finding the time-window of the closest collaboration.}
    \label{fig:case-study-2}
\end{figure}



\hspace{-1em}\textbf{\textit{Challenges and Contributions}}
Counting butterflies in such a historical setting is challenging. One main reason is that the algorithm should be able to answer historical queries multiple times to analyze the changing trend.
When the graph is large, it is inefficient to run existing butterfly counting algorithms from scratch for each query. Thus, we need to design algorithms that answer each query more efficiently after preprocessing. No algorithm can effectively solve this problem in existing works. This paper fills this hole by proposing a new index algorithm with consistently high performance in large-scale graphs. The proposed algorithm \texttt{GSI} (graph structure-aware index) can take advantage of graph structures and balance between query time and memory usage, enabling it to outperform previous butterfly counting algorithms on both real-world data, and synthetic data with certain distributions. \texttt{GSI} can also be parallelized, providing faster query efficiency. It can also be compressed, providing smaller memory usage. To summarize, we have made the following contributions.
\begin{itemize}[leftmargin=*]
    \item We introduce the historical butterfly counting problem and prove its hardness. This enables in-depth trend and dynamics analysis over temporal graphs, and helps understand how temporal variations influence network structures over time.
    \item The graph structure-aware index (\texttt{GSI}) is designed to support efficient counting query with a controllable balance between query time and memory usage. Theoretically, we prove that the \texttt{GSI} approach transcends conventional computational complexity barriers associated with general graphs when applied to power-law graphs, a common characteristic of many real-world graphs.
    \item When \modify{memory becomes the bottleneck and} exact counting is not required, we propose an index compression strategy to provide fast, high-quality, and unbiased approximations of the counts based on the compressed index.
    \item Extensive experiments demonstrate that our query algorithm achieves up to five orders of magnitude speedup over the state-of-the-art solutions with manageable memory footprints.
\end{itemize}






\vspace{-1.6em}
\section{RELATED WORK}
\label{sec:related}

In this section, we review the related works, including the butterfly counting on static bipartite graphs, motif counting on temporal graphs, and other historical queries on temporal graphs.

$\bullet$ \textbf{Butterfly Counting on Static Bipartite Graphs.} Butterfly is the most fundamental sub-structure in bipartite graphs.
Significant research efforts have been dedicated to the study of counting and enumerating butterflies on static bipartite graphs \cite{wang2019vertex,sanei2018butterfly, wang2014rectangle}.
Wang et al. \cite{wang2014rectangle} first proposed the butterfly counting problem and designed an algorithm by enumerating wedges from a randomly selected layer.
Sanei-Mehri et al. \cite{sanei2018butterfly} further developed a strategy for choosing the layer to obtain better performance.
Recently, Wang et al. \cite{wang2019vertex} utilized the vertex priority and cache optimization to achieve state-of-the-art efficiency.
Additionally, the parallel algorithms \cite{sanei2018butterfly,shi2022parallel}, I/O efficient algorithm \cite{wang2023efficient}, sampling-based algorithms \cite{sanei2018butterfly,li2021approximately,sheshbolouki2022sgrapp}, GPU-based algorithm \cite{sheshbolouki2022sgrapp,xia2024gpu}, and batch update algorithm \cite{wang2023accelerated} have also been developed for the butterfly counting problem.
In addition, the butterfly counting problems in steam, uncertain, and temporal bipartite graphs have also been studied~\cite{cai2024efficient,sheshbolouki2022sgrapp,sanei2019fleet, zhou2021butterfly}.
%

$\bullet$ \textbf{Motif Counting on Temporal Graphs.} The problem of temporal motif counting has been extensively studied recently \cite{boekhout2019efficiently,gurukar2015commit,li2018temporal,liu2021temporal,kovanen2011temporal}.
Kovanen et al. \cite{kovanen2011temporal} introduced the concept of $\Delta$-adjacency, which pertains to two temporal edges sharing a vertex and having a timestamp
difference of at most  $\Delta$, and consider the temporal ordering aspect. 
The $\Delta$-temporal motif counting with \cite{pashanasangi2021faster,redmond2013temporal} and without temporal ordering \cite{paranjape2017motifs} are also been studied. 
Furthermore, there are numerous approximation algorithms available for solving counting problems \cite{liu2019sampling,sarpe2021oden}. When it comes to enumeration problems, isomorphism-based algorithms are
the most commonly used \cite{li2019time,locicero2021temporalri}.

$\bullet$ \textbf{Other Historical Queries on Temporal Graphs.} The historical queries on the temporal graphs aim to compute the specific structure in the snapshot of an arbitrary
time window.
The historical reachability~\cite{wen2022span}, $k$-core~\cite{yu2021querying}, structural diversity~\cite{chenquerying}, and connected components~\cite{xie2023querying} of temporal graphs have been defined, and index-based solutions have also been proposed.
Note that our work is significantly different with \cite{cai2024efficient}, as our work is focused on more generalized scenarios of temporal graph mining, i.e., analyzing the relationship between vertices in a time interval and without any temporal ordering limitations.

\section{Preliminaries}
\label{sec:prelim}

Throughout the paper, we use the big-O notation to express the upper bounds of the \textit{absolute value}s of functions. For example, $O(\frac{1}{n})$ includes functions between $-\frac{c}{n}$ and $\frac{c}{n}$ for some constant $c$ and sufficiently large $n$. We say $f(n) \ge g(n)+O(h(n))$ if $f(n)\ge g(n)-c|h(n)|$ for some constant $c$ and sufficiently large $n$. 
We use $\widetilde{O}(f(n))$ to denote all functions bounded by $O(f(n)\log^cn)$ for some constant $c.$ 
We use $[n]$ to denote the set $\{1, \ldots, n\}.$
For an undirected graph, we denote every vertex $u$'s degree as $deg_u$.


\subsection{Problem Definitions}
\begin{figure}[t!]
    \centering
    \includegraphics[width=1\linewidth]{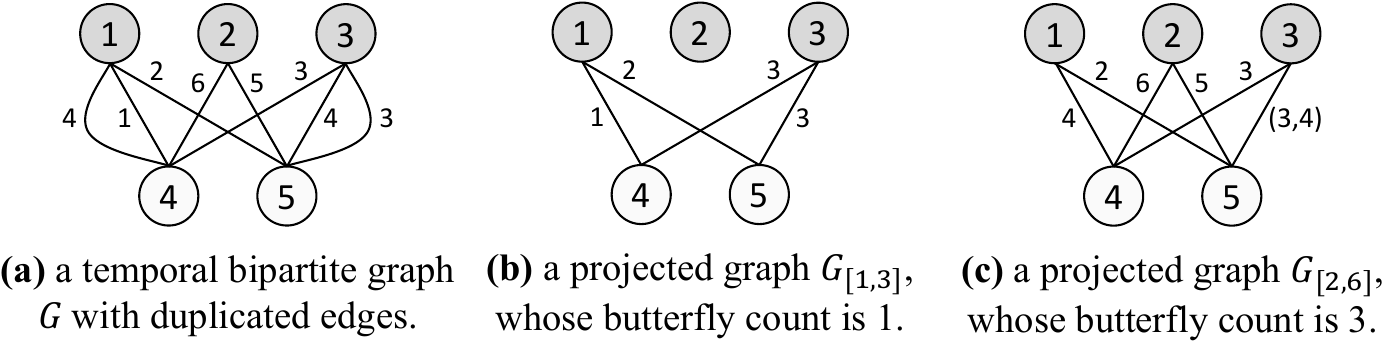}
    \caption{A temporal bipartite graph and its projected graphs in two time-windows, associated with their butterfly counts.}
    \label{fig:TG-PG-BC}
\end{figure}

Two common motifs are being widely studied on bipartite graphs: wedges and butterflies. We give their formal definition as follows:

\begin{mydef}[Wedge~\cite{wang2019vertex}]
Given a bipartite graph $G = (V = (U, L), E)$, a wedge $\left \langle x \leadsto y \leadsto z \right \rangle$ is a 2-hop path consisting of edges $(x, y)$ and $(y, z)$.
\end{mydef}

\begin{mydef}[Butterfly~\cite{wang2019vertex}]

Given a bipartite graph $G = (V = (U, L), E)$, and the four vertices $x, y, z, w \in V$ where $x, z \in U$ and $y, w \in L$, $b: \left \langle x, y, z, w \right \rangle$ is a butterfly iff the subgraph induced by $x, y, z, w$ is a $(2,2)$-biclique of $G$; that is, $x$ and $z$ are all connected to $y$ and $w$, respectively.
\end{mydef}

Throughout the paper, we study the temporal bipartite graphs. A temporal bipartite graph is an undirected graph $G = (V = (U, L), E)$, where each edge $e \in E$ is a triple $(u, v, t)$ with two vertices $u \in U, v \in L$ and a timestamp $t$. Our major focus is the historical type query on temporal bipartite graphs; that is, we query on an extracted graph from $G$ with respect to a certain time window. We denote it as the \textit{projected graph} and its formal definition is as follows:

\begin{mydef}[Projected graph~\cite{fang2020survey}]\label{def:projected-graph}
Given a temporal bipartite graph $G = (V, E)$ and a time-window $[t_s, t_e] (t_s \leq t_e)$, the projected graph $G_{[t_s, t_e]}$ of $G$ is an undirected bipartite graph without timestamps, where its vertex set $V_{[t_s, t_e]}$ is $V$ and the edge set $E_{[t_s,t_e]}$ is $\{(u, v)\ |\ \exists (u,v,t)\in E\  \land\ t \in [t_s, t_e]\}$.
\end{mydef}

Now, we are ready to state the major problem formally:

\begin{problem}[Historical butterfly counting]
\label{problem:main}
Given a temporal bipartite graph $G$ and a time-window $[t_s, t_e]$, find the number of butterflies in the projected graph $G_{[t_s, t_e]}$.
\end{problem}

In \Cref{fig:TG-PG-BC}, we are given a temporal bipartite graph $G$ in (a), in which the number represents the timestamp of each edge. We consider two time-windows $[1,3]$ and $[2,6]$. The corresponding projected graphs $G_{[1,3]}$ and $G_{[2,6]}$ are shown in (b) and (c), respectively. The historical butterfly counting query associated with these two-time windows is, in fact, counting on these two projected graphs. Therefore, the answer is $1$ ($\left \langle 1, 3, 4, 5 \right \rangle$) for $[1,3]$ and $3$ ($\left \langle 1, 2, 4, 5 \right \rangle$, $\left \langle 1, 3, 4, 5 \right \rangle$, $\left \langle 2, 3, 4, 5 \right \rangle$) for $[2,6]$.

\subsection{Some Key Techniques in Motif Counting}
In the previous works on motif counting~\cite{wang2019vertex,wang2023efficient}, the value of $\frac{1}{|E|} \sum_{(u, v) \in E} \min(\text{deg}_u, \text{deg}_v)$ is widely used for complexity analysis for the given input graph $G = (V, E)$.
By \cite{chiba1985arboricity}, this can be simplified as $O(\delta)$, where $\delta$ is defined as the arboricity of the given graph.
We will adapt the $\delta$ notion in our paper for brevity.

Our proposed algorithms widely use two important techniques: the \textit{Vertex Priority} method and the \textit{Chazelle’s structure}. 

\paragraph{Vertex Priority} The vertex priority method reduces the number of wedges we need to consider. To begin with, we define the vertex priority as follows:

\begin{mydef}[Vertex priority~\cite{wang2019vertex}] \label{def:priority}
For any pair vertices $x, y$ in a temporal bipartite graph $G$, we define $x$ is prior to $y$ ($pr(x) \prec pr(y)$) if and only if: $\overline{deg}_x$ > $\overline{deg}_y$, or $\overline{deg}_x$ = $\overline{deg}_y$ and $id(x) < id(y)$,
where $\overline{deg}_u$ denotes the degree of $u$ when only considering unique edges in $G$, and $id(u)$ denotes the unique ID of $u$.
\end{mydef}

By ~\cite{wang2019vertex}, we only need to consider the wedges $\left \langle x \leadsto y \leadsto z \right \rangle$ that satisfy $pr(x) \prec pr(y) \land pr(x) \prec pr(z)$ in order to count the butterflies without repetition or missing. Each butterfly $\left \langle x, y, z, w\right \rangle$ is constructed from two such wedges $\left \langle x \leadsto y \leadsto z \right \rangle$ and $\left \langle x \leadsto w \leadsto z \right \rangle.$

\paragraph{2D-range Counting} Our algorithm will transform counting butterflies into counting points on a 2-dimensional plane, known as the \textit{2D-range counting} problem. Specifically, let $P$ be a set of $n$ points in 2-d space $\mathbb{R}^{2}$. The \textit{2D-range counting} problem is: Given an orthogonal rectangle $Q$ of the form $[x_1, x_2] \times [y_1, y_2]$, find the size of $|Q \cap P|$.
For an instance of the 2D-range counting problem, we use a classic data structure known as Chazelle’s structure~\cite{chazelle1988functional} to handle all 2D-range counting queries after preprocessing:

\begin{theorem}[Chazelle’s structure~\cite{chazelle1988functional}]
\label{theo:cs}
A Chazelle’s structure $\mathcal{CS}$ is a data structure that can answer each 2D-range counting in $O(\log n)$ time and $O(\frac{n\log n}{\omega})$ memory usage, where $\omega$ is the word size. The preprocessing time is $O(n \log n)$.
\end{theorem}

In practice, $2^{\omega}$ significantly exceeds the number of points involved in our method's 2D-range counting task. Therefore, we assume that $\log n$ is $O(\omega)$ for Chazelle’s structures used and simplify the memory usage into $O(n)$ for brevity. In later sections, we use $\mathcal{CS}$ to denote such a data structure.


\subsection{Baselines}
Our baseline solutions are built upon the \textit{states-of-the-arts} methods for exact butterfly counting (i.e., \texttt{BFC-VP++}~\cite{wang2019vertex}) and approximate butterfly counting (i.e., Weighted Pair Sampling (\texttt{WPS})~\cite{zhang2023scalable}) on the static graphs.
We need to extract the static graph from the temporal bipartite graph for a given time-windows first, and then run the following solutions:
\begin{itemize}[leftmargin=*]
    \item \texttt{BFC-VP++}: The \texttt{BFC-VP++} algorithm sorts vertices by their proposed vertex priority (\Cref{def:priority}) and efficiently identifies almost all minimally redundant wedges that can form a butterfly, offering a practical, efficient, and cache-friendly solution.
    In addition, \texttt{BFC-VP++} can be highly parallelized. We will also compare our methods with its parallel version in the following.
    \item \texttt{WPS}: The basic idea behind \texttt{WPS} is to estimate the total butterfly count with the number of butterflies containing two randomly sampled vertices from the same side of the two vertex sets. The method has been proven as unbiased and theoretically efficient in power-law graph models.
\end{itemize}
\section{INDEX-BASED ALGORITHMS}
\label{sec:algo}
In this section,  we consider solving \Cref{problem:main} exactly by index-based solutions. To begin with, in \Cref{sec:hardness}, we give a hardness result showing that $\widetilde{O}(m^{2}/\lambda^2)$ space is needed to answer queries exactly in $\widetilde{O}(\lambda)$ time, where $m$ denotes the number of edges in the given temporal bipartite graph $G$. 
Correspondingly, we provide an algorithm that meets such bound in \Cref{sec:ebi}, named as \textit{Enumeration-based Index (\texttt{EBI})}.
\texttt{EBI} achieves $\Otil(1)$ query time but it needs expensive memory usage for large graph data.
In \Cref{sec:cbi}, we introduce a different algorithm named \textit{Combination-based Index (\texttt{CBI})} that can be constructed under practical memory constraints.

To bridge the gap between theory and practice, we propose a new algorithm in \Cref{sec:gsi} that effectively combines \texttt{EBI} and CBI, demonstrating strong performance on real-world graph data. This algorithm, named \textit{Graph Structure-aware Index (\texttt{GSI})}, intelligently allocates graph data to the two indices which allows it to take advantage of the underlying graph structure.
In addition, without compromising the performance, \texttt{GSI} can also handle duplicate edges with proper modification, which is discussed in \modify{the technical report}.



\subsection{Problem Hardness}

Our result is motivated by the hardness result in \cite{deng2023space} where they reduce the \textsc{Set Disjointness} problem (\Cref{def:setdisjoint}) into \Cref{problem:main}. While \cite{deng2023space} considers timestamps on vertices, we prove for the setting where timestamps are on edges separately.
\begin{mydef} [Set Disjointness]
\label{def:setdisjoint}
Given a collection of $s\ge 2$ sets $S_1$, $S_2$, $\ldots$,$S_s$,  a query $(a, b)\in [s]^2$ asks for whether $S_a\cap S_b$ is empty.
\end{mydef}

The strong set disjointness conjecture~\cite{Goldstein17,Goldstein19} is as follows:
\begin{theorem}
    For the set disjointness problem, any data structure with query time $\lambda$ must use $\widetilde{\Omega} (n^2/\lambda^2)$ space where $n$ is the sum of the sizes of $S_1,\ldots, S_s$.
\end{theorem}

Specifically, we prove the following result for \Cref{problem:main}:

\begin{theorem}
\label{thm:hardness}
Consider \Cref{problem:main}. Let $m$ denote the number of edges in $G$. Fix any $\lambda \in [1,m]$ and $\delta>0$. Suppose that we have a data structure for \Cref{problem:main} using $\widetilde{O}(m^{2-\delta}/\lambda^2)$ space and exactly answers each query in $\widetilde{O}(\lambda)$ time. Then, for any set disjointness problem with $\sum_{i=1}^s |S_s|=N$, we have a data structure that uses $\widetilde{O}(N^{2-\delta}/\lambda^2)$ time and that +answers each query in $\widetilde{O}(\lambda)$ time.
\end{theorem}
\begin{proof}
    Suppose that we have a data structure for \Cref{problem:main} using $\widetilde{O}(m^{2-\delta}/\lambda^2)$ space and exactly answers each query in $\widetilde{O}(\lambda)$ time. We can utilize the data structure for a set disjointness instance as follows: 
    
    Let there be $s$ sets $S_1,\ldots, S_s$. Without loss of generality, we let $1,\ldots, t$ denote all distinct elements in $S_1,\ldots, S_s$, \textit{i.e.,} $\bigcup_{i=1}^s S_i=[t]$. We construct a graph $G=(V=(U, L), E)$ such that $U=[s]$ and $L=[t+1]$. An edge $(i, j)$ exists if $j\in S_i$ or $j=t+1$. We assign timestamp $i$ to the edge $(i, j)$.

    We build the data structure for the graph $G$ defined above. For two integers $a \le b \in [s]$, we can query the data structure for the number of butterflies in the time window $[a, b]$. Since the timestamp on edge $(i, j)$ is equal to $i$, such butterflies are exactly those with the form $\left< u, v, w, x\right>$ such that $u,w\in [a,b]$. The number of such butterflies is 
    \[
    Q(a, b)\defeq \sum_{i=a}^b \sum_{j=i+1}^b \left(\begin{array}{c}|S_i\cap S_j|+1 \\ 2\end{array}\right).
    \]
    Notice that $Q(a, b)$ is the two-dimensional prefix sum of $\binom{|S_a\cap S_b|+1}{2}$. In other words, we can calculate $\binom{|S_a\cap S_b|+1}{2}$ by 
    $Q(a, b)-Q(a+1, b)-Q(a, b-1)+Q(a+1, b-1)$,
    where we define $Q(x, y)=0$ for $x\ge y$.
    We can answer whether $|S_a\cap S_b|=0$ by checking whether this value is equal to zero.
    
\end{proof}

\label{sec:hardness}
\begin{figure}[t!]
    \centering
    \includegraphics[width=0.9\linewidth]{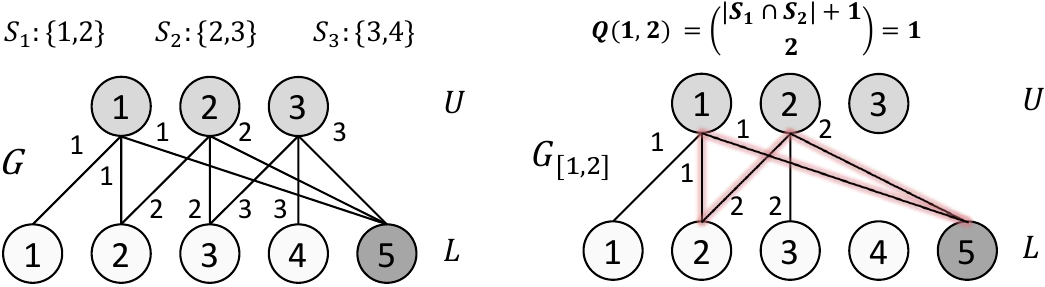}
    \caption{An example for reducing the set disjointness problem into exact historical butterfly counting in temporal bipartite graphs. }
    \label{fig:PH}
\end{figure}

For example, in the \Cref{fig:PH}, we reduce the set disjointness problem for three sets $\{1, 2\}$, $\{2, 3\}$ and $\{3, 4\}$ to a historical butterfly counting problem on the graph we constructed.
For determining whether $S_1\cap S_2$ is empty, we count the number of butterflies of the form $\left<1, v, 2, x\right>$, $Q(1, 2),$ which is equal to $\binom{\left|S_1\cap S_2\right| + 1}{2}$. 
For determining whether $S_1\cap S_3$ is empty, we calculate $\binom{|S_1\cap S_3|+1}{2}=Q(1, 3) - Q(2. 3) - Q(1, 2) + Q(2, 2)=2 - 1 - 1 + 0 = 0$.

This hardness result shows that to achieve efficient query runtime (e.g., $\widetilde{O}(\lambda)$), large memory space is necessary (e.g., $\widetilde{O}(m^2/\lambda^2)$). Even though we can design theoretically optimal algorithms that reach the lower bound, they might not be practical when the input graph is large. One possible way to overcome such a challenge from a practical perspective is to develop algorithms that take advantage of real-world graphs' properties.

\subsection{Enumeration-based Index}
\label{sec:ebi}
Since the structure of a temporal graph changes over time, each subgraph has its own life cycle, i.e., the only interval of time in which it exists. 
We denote such an interval as the active timestamp of this subgraph. We give a formal definition as follows:
\begin{mydef}[Active Timestamp]
\label{def:original-timestamp}
For any subgraph $P$ of a bipartite temporal graph $G$, we define the active timestamp $\mathcal{T}(P)$ as a pair of ordered timestamp $[l, r] (l \leq r)$, such that $P$ exist in the projected graph $G_{[t_s, t_e]}$ if and only if $t_s \leq l$ and $r \leq t_e$.
\end{mydef}

\begin{figure}[t!]
    \centering
    \includegraphics[width=1\linewidth]{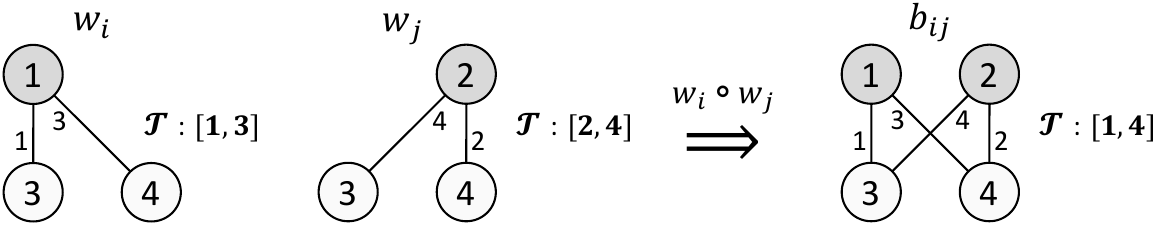}
    \caption{An illustrative example for active timestamps of wedges and the butterfly constructed from them.}
    \label{fig:timestamp}
\end{figure}

In \Cref{fig:timestamp}, there are two wedges $w_i$ and $w_j$. For $w_i$, since it contains two edges $(1,3,1)$ and $(1,4,3)$, it only exists in a projected graph $G_{[t_s, t_e]}$ satisfying $t_s \leq 1$ and $t_e \geq 3$. The active timestamp $\mathcal{T}(w_i)$ is $[1,3]$. Similarly, for $w_j$, its active timestamp is $[2,4]$.

The proposed \textit{Enumeration-based Index (\texttt{EBI})} algorithm is shown in \Cref{alg:EBI-Construction} and \Cref{alg:EBI-Querying}. For brevity, we assume that no duplicate edge exists. We discuss how to handle duplicate edges without compromising the performance in \modify{the technical report}.


$\bullet$ \textbf{Construction} The goal is to enumerate every butterfly and compute its active timestamp. After initializing a 2D-range counting data structure $T_E$ by all butterflies' timestamps, we can answer every historical butterfly counting query through a single query on $T_E$. We initialize it to be empty in the beginning (\Cref{ebi-c:initds}).

To find all the butterflies in $G$, we enumerate all the wedges. We consider two wedges $w_i, w_j$ of $G$ with shared endpoints but different middle vertices, denoted as $w_i: \left \langle x \leadsto y_i \leadsto z \right \rangle$ and $w_j : \left \langle x \leadsto y_j \leadsto z \right \rangle$. Let $\mathcal{T}(w_i)=[l_i,r_i]$ and $\mathcal{T}(w_j)=[l_j,r_j]$. The temporal butterfly constructed from them is denoted as $b_{ij} = w_i \circ w_j$, where $\mathcal{T}(b_{ij})=[\min(l_i, l_j), \max(r_i, r_j)]$. 

In \Cref{fig:timestamp}, we demonstrate how to construct a butterfly $b_{ij}$ with two wedges $w_i$ and $w_j$. Since $b_{ij}$ contains edges $(1,3,1)$, $(1,4,3)$, $(2,3,4)$, and $(2,4,2)$, its active timestamp can be computed by $l = \min \{t |\  \exists e = (u, v, t) \in edges(b_{ij})\}=1,$ and $r = \max \{t |\  \exists e = (u, v, t) \in b_{ij}\}=4$. By the definition of the $\circ$ operation, it can also be computed from $\mathcal{T}(w_i)$ and $\mathcal{T}(w_j)$. This also indicates that a butterfly $b_{ij} = w_{j} \circ w_{j}$ exists in a given time-window, if and only if the active timestamp of $w_i$ and $w_j$ are both included in the window. Such observation helps us develop an alternative index approach in the later section.

We follow the method in \cite{wang2019vertex} to enumerate every wedge $w_i : \left \langle x \leadsto y \leadsto z \right \rangle [l_i, r_i]$ in $G$ without repetition or missing (\Cref{ebi-c:enumerate1}). The complexity is bounded by the total number of wedges in $G$. Let $W[(x, z)]$ be a set of wedges whose endpoints are $x, z$. $W$ is the family of sets $W[(x, z)]$. $W$ is initialized to be empty (\Cref{ebi-c:initw}). For the current wedge $w_i$, we first check if its endpoints $(x, z)$ exists in $W$ (\Cref{ebi-c:if}) (by a hash table, for example). If not, we initialize the corresponding set $W[(x, z)]$ to be empty (\Cref{ebi-c:initbucket}).

Since any two different wedges with the same pair of endpoints $(x,z)$ construct a unique butterfly, we enumerate every $w_j$ in \\ $W[(x, z)]$ (\Cref{ebi-c:enumerate2}) and construct a butterfly $b_{ij}=w_i \circ w_j$ with $\mathcal{T}(b_{ij})=(t_l,t_r)$ (\Cref{ebi-c:butterfly}). Correspondingly, we insert a single 2D point $(t_l,t_r)$ into $T_E$. After enumeration, we need to update $W[(x, z)]$ by inserting $w_i$ in it (\Cref{ebi-c:update}). After finding all the butterflies, we run the preprocessing for $T_E$ and return it (\Cref{ebi-c:return}).

In \Cref{fig:ebi}, we provide an example of our construction process. There are six wedges $w_1: \left \langle 1 \leadsto 4 \leadsto 2 \right \rangle$, $w_2: \left \langle 1 \leadsto 5 \leadsto 2 \right \rangle$, $w_3: \left \langle 1 \leadsto 4 \leadsto 3 \right \rangle$, $w_4: \left \langle 1 \leadsto 5 \leadsto 3 \right \rangle$, $w_5: \left \langle 2 \leadsto 4 \leadsto 3 \right \rangle$, $w_6: \left \langle 2 \leadsto 5 \leadsto 3 \right \rangle$ where $\mathcal{T}(w_1)=[4,6]$, $\mathcal{T}(w_2)=[2,5]$, $\mathcal{T}(w_3)=[4,4]$, $\mathcal{T}(w_4)=[2,4]$, $\mathcal{T}(w_5)=[4,6]$, $\mathcal{T}(w_6)=[4,5]$. We group them by $(x,z)$. There will be three groups (sets) of wedges: $W[(1,2)]=\{w_1,w_2\}$, $W[(1,3)]=\{w_3,w_4\}$, $W[(2,3)]=\{w_5,w_6\}$. For each group, we construct a butterfly $b_{ij}$ from each pair of different wedges $w_i,w_j$ from the group. We have $b_1=w_1 \circ w_2$, $b_2=w_3 \circ w_4$, $b_3=w_5 \circ w_6$. We compute their active timestamp through simple calculation: $\mathcal{T}(b_1)=[2,4]$, $\mathcal{T}(b_2)=[2,6]$, $\mathcal{T}(b_3)=[4,6]$. In the end, we insert $(2,4)$, $(2,6)$, $(4,6)$ into $T_E$. Note that in this showcase, for demonstrating the group dividing, we do not follow the vertex priority in \Cref{def:priority}. Otherwise, we only consider three wedges and they are all in the group $W[(4,5)]$.


$\bullet$ \textbf{Answering Query} Since each butterfly with active timestamp $[l,r]$ is represented by a point $(l, r)$ in $T_E$, a historical butterfly counting of time-window $[t_s,t_e]$ can be interpreted as a 2D-range query counting the number of points in the rectangle $[t_s, \infty] \times [-\infty, t_e]$. We query $T_E$ for it and return the answer directly (\Cref{ebi-q:return}). 

In \Cref{fig:ebi}, after the preprocessing of $T_E$, we can answer any 2D-range counting query concerning those three points, which represent the three butterflies. When we want to query the number of butterflies in time-window $[1,5]$, we are asking how many butterflies' active timestamps are in the range of $[1,5]$. We interpret this as querying the number of points in the rectangle $[1,\infty] \times [-\infty,5]$. There is only one point $(2,4)$ in the rectangle. Therefore, there is only one butterfly in the projected graph of the queried time-window. The calculation is done by a single query on $T_E$. It is a typical 2D-range counting query.

$\bullet$ \textbf{Time and Space Complexity} The time and space complexity of \texttt{EBI} is summarized as follows: 
\begin{theorem}
\label{theo:EBI}
    Denote $B$ as the number of butterflies of the input graph $G$ and $\delta$ as the arboricity, \texttt{EBI}'s construction (\Cref{alg:EBI-Construction}) runs in $O(m\delta + B \log m)$ and takes $O(B)$ space. After the construction, each query (\Cref{alg:EBI-Querying}) takes $O(\log m)$ to answer. 
\end{theorem}
\begin{proof}
The construction process requires enumerating all wedges. By \cite{wang2019vertex}, this takes $O(m\delta)$. For each wedge $w_i : \left \langle x \leadsto y \leadsto z \right \rangle$, it is inserted in $W[(x,z)]$. By enumerating every existing $w_j$ in $W[(x,z)]$, every butterfly $b_{ij}=w_i \circ w_j$ is enumerated exactly once. Since the $\circ$ operation takes $O(1)$ and the insertion to $T_E$ takes $O(\log m)$ by \Cref{theo:cs}, the total preprocess time for all butterflies takes $O(B\log m)$. Therefore, the total runtime for \Cref{alg:EBI-Construction} is $O(m\delta + B\log m)$ as desired. Every wedge will be stored in $W[(x, z)]$ where $x, z$ are its endpoints. All the wedges take $O(m\delta)$ space in total. $T_E$ contains $O(B)$ points in total. By \Cref{theo:cs}, it takes $O(B)$ space. The total memory usage sum up to $O(B+m\delta)$. Each query on \texttt{EBI} is a range query on $T_E$. By \Cref{theo:cs}, this takes $O(\log m)$ time and no additional space.
\end{proof}
The \texttt{EBI} algorithm reaches the bound in \Cref{sec:hardness}: 
Since $\delta$ is bounded by $O(\sqrt{m})$ and $B$ is bounded by $O(m^2)$, the memory usage is bounded by $\Otil(m^2)$.
Therefore, \texttt{EBI} answers each query in $\widetilde{O}(1)$ time and take $\Otil(m^2)$ space. 
This indicates that \texttt{EBI}'s asymptotic memory usage cannot be further improved without affecting the $\widetilde{O}(1)$ query time. On the other hand, though reaching theoretical optimality, \texttt{EBI}'s performance on large-scale data is not ideal because enumerating and storing all butterflies is challenging for these graphs. For example, in the \texttt{Wiktionary} dataset\footnote{http://konect.cc/}, the butterfly count exceeds $10^{16}$ where the graph contains about $4 \times 10^{7}$ edges\footnote{Here, \#butterflies is larger than the square of \#edges due to duplicate edges, which will be handled in the technical report.}. 
\begin{figure}[t!]
    \centering
    \includegraphics[width=1\linewidth]{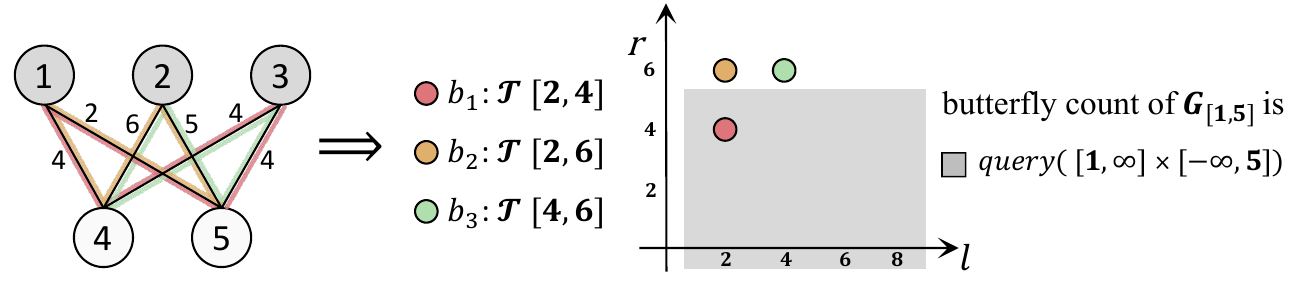}
    \caption{An example for transforming historical butterfly counting into 2D-range counting.}
    \label{fig:ebi}
\end{figure}



\begin{algorithm}[t]
\small
\caption{\small \textsc{EBI-Construction}}\label{alg:EBI-Construction}
\KwIn{a bipartite temporal graph: $G = (V, E)$; }
\KwOut{a $\mathcal{CS}$ (\Cref{theo:cs}) for indexing (EB-index): $T_{E}$;}
Initialize $T_{E}$ with an empty $\mathcal{CS}$; \\ \label{ebi-c:initds}
$W \leftarrow$ an empty hashmap mapping pairs of vertices to sets; \\ \label{ebi-c:initw}
\For {each wedge $w_i : \left \langle x \leadsto y \leadsto z \right \rangle [l_i, r_i] \in G$}{ \label{ebi-c:enumerate1}
            \lIf{$(x, z) \notin W.keys()$}{ \label{ebi-c:if}
                $W[(x, z)] \leftarrow \emptyset$ \label{ebi-c:initbucket}
            }
            \For{$w_j \in W[(x, z)]$}{ \label{ebi-c:enumerate2}
                $b_{ij} \leftarrow w_i \circ w_j$; \\ \label{ebi-c:butterfly}
                $T_E.insert(t_l, t_r)$\tcp*{$(t_{l}, t_{r})=\mathcal{T}(b_{ij})$} \label{ebi-c:insert}
                
            }
            $W[(x, z)] \leftarrow W[(x, z)] \cup w_i$; \label{ebi-c:update}
}
Preprocess $T_E$ (\Cref{theo:cs}); \\
\Return $T_E$; \label{ebi-c:return}
\end{algorithm}

\begin{algorithm}[t]
\small
\caption{\small \textsc{EBI-Querying}}\label{alg:EBI-Querying}
\KwIn{a time-window: $[t_s, t_e]$; a EB-index: $T_{E}$;}
\KwOut{temporal butterfly counting: $num_{[t_s, t_e]}$;}
\Return $num_{[t_s, t_e]} \leftarrow T_{E}.query([t_s, \infty] \times [-\infty, t_e])$; \label{ebi-q:return}
\end{algorithm}

\subsection{Combination-based Index}
If the graph is large, it is not realistic to pre-compute and store all butterflies to handle historical queries, as this requires too much space. On the other hand, we are only interested in the number of butterflies, so it is not necessary to construct them explicitly. 

Recall that a butterfly exists in a given time window if and only if the active timestamp of the two wedges are both included in the window. Consider a \textit{group of wedges} $S=\{\left \langle x \leadsto y \leadsto z \right \rangle\}$ where $x,z$ are fixed. Any two wedges from $S$ form a butterfly. Therefore, the total number of butterflies that comes from $S$ is $\binom{|S|}{2}$. This idea leads us to a different index algorithm, named \textit{Combination-based Index (\texttt{CBI})}. The implementation details are provided in \Cref{alg:CBI-Construction} and \Cref{alg:CBI-Querying}. Similarly, we assume no duplicate edge for brevity. Resolving them is deferred to \modify{the technical report}.


$\bullet$ \textbf{Construction} In \texttt{EBI}, we use one 2D-counting data structure to store butterflies. Instead, in \texttt{CBI}, we map every type of wedge to a separate data structure. We initialize a hashmap $T_C$ in the beginning (\Cref{cbi-c:initds}). We still need to enumerate every wedge $w_i : \left \langle x \leadsto y \leadsto z \right \rangle [l_i, r_i]$ in $G$ (\Cref{cbi-c:enumerate}). A wedge with endpoints $x, z$ is maintained by the 2D-counting data structure $T_C[(x, z)]$. If $T_C[(x, z)]$ does not exist(\Cref{cbi-c:if}), we initialize an empty $\mathcal{CS}$ (\Cref{theo:cs}) for $T_C[(x,z)]$ (\Cref{cbi-c:init}). Unlike in \texttt{EBI}, where we compute every butterfly that is constructed from the new wedges $w_i$ and a previous wedge $w_j \in W[(x, z)]$, we directly insert a point $[l_i,r_i]$ ($w_i$'s active timestamp) into the corresponding $T_C[(x,z)]$ (\Cref{cbi-c:insert}). After the enumeration, we preprocess all data structures $T_C[(x, z)]$ and return $T_C$ (\Cref{cbi-c:return}).

$\bullet$ \textbf{Answering Query} To answer a historical butterfly counting of time-window $[t_s,t_e]$, for every type of wedges, we first need to compute the number of wedges that exist in the time-window, then compute the number of butterflies by a simple binomial coefficient. We set a counter $num_{[t_s, t_e]}$ to be $0$ initially (\Cref{cbi-q:init}). We enumerate every existing data structure in $T_C$ (\Cref{cbi-q:enumerate}). Similar to what we do in \texttt{EBI}, we query on the corresponding 2D-counting data structure $T_C[(x,z)]$ with a query $[t_s, \infty] \times [-\infty, t_e]$ to get the number of wedges that are active in the time-window (\Cref{cbi-q:query}). Denoting the query answer as $count$, we add $\binom{count}{2}$ to $num_{[t_s, t_e]}$, which is the number of butterflies that are constructed from these wedges. After enumerating all the data structures, we return $num_{[t_s, t_e]}$ as the total number of butterflies (\Cref{cbi-c:return}).

$\bullet$ \textbf{Time and Space Complexity} The time and space complexity of \texttt{CBI} is summarized as follows:
\begin{theorem}
\label{theo:CBI}
Denote $\widetilde{w}$ as the number of wedge groups and $\delta$ as the arboricity. 
\texttt{CBI}'s construction (\Cref{alg:CBI-Construction}) runs in $O(m\delta\log n)$ and takes $O(m\delta)$ space. After the construction, \texttt{CBI} answers each query (\Cref{alg:CBI-Querying}) in $O(\widetilde{w} \log n)$ time.     
\end{theorem}

\begin{proof}
    The construction process (\Cref{alg:CBI-Construction}) also requires enumerating all wedges, which takes $O(m\delta)$ by \cite{wang2019vertex}. For each group of wedges with the same endpoints, a $\mathcal{CS}$ needs to be initialized with all wedges from this group. Since the number of wedges is bounded by $O(m\delta)$ and each group contains at most $O(n)$ wedges, the total time for preprocessing all the data structures is $O(m\delta\log n)$ by \Cref{theo:cs}. Similar to \Cref{theo:EBI}, it takes $O(m\delta)$ memory in total. When answering a query (\Cref{alg:CBI-Querying}), all groups of wedges need to be enumerated, and their corresponding $\mathcal{CS}$ will be queried exactly once. The total query time will be $O(\widetilde{w}\log n)=\widetilde{w} \times O(\log n)$ by \Cref{theo:cs}. Memory usage is dominated by 2D-range counting data structures for each group of wedges. 
\end{proof}

As a result, \texttt{CBI} manages to resolve the memory issue of \texttt{EBI} when there are too many butterflies. Although $\widetilde{w}$ is bounded by the number of wedges (\textit{i.e.,} $m\delta$), in practice, it is significantly smaller (\textit{i.e.,} $\widetilde{w} \ll m\delta$). For example, in the \texttt{Wiktionary} dataset, there is about $m \delta \approx 1.3 \times 10^{9}$ wedges but only $\widetilde{w} \approx 5 \times 10^{7}$ wedge groups with different $(x, z)$.  

By further observing the distribution of wedges on real-world data, we see a concentration phenomenon that our algorithm has not addressed: Most wedges are grouped in a few groups, while the rest contain only a very small number of wedges. Therefore, building a separate data structure for every wedge group might be too expensive and unnecessary, especially for groups that only contribute a small amount of butterflies. For example, in the \texttt{Wiktionary (WT)} dataset, more than $81.3$\% wedge groups contain less than 10 wedges.

\label{sec:cbi}
\begin{algorithm}[t]
\small
\caption{\small \textsc{CBI-Construction}}\label{alg:CBI-Construction}
\KwIn{a bipartite temporal graph: $G = (V, E)$; }
\KwOut{a hashmap mapping pairs of vertices to multiple $\mathcal{CS}$s: $T_{C}$;}

$T_C \leftarrow$ an empty hashmap mapping pairs of vertices to $\mathcal{CS}$s; \\\label{cbi-c:initds}
\For{each wedge $w_i : \left \langle x \leadsto y \leadsto z \right \rangle [l_i, r_i] \in G$}{\label{cbi-c:enumerate}
    \If{$(x, z) \notin T_{C}.keys()$}{ \label{cbi-c:if}
        Initialize $T_{C}[(x, z)]$ with an empty $\mathcal{CS}$; \label{cbi-c:init}
    }
    $T_{C}[(x, z)].insert(l_i, r_i)$; \\ \label{cbi-c:insert}
}
Preprocess each data structure in $T_C$ (\Cref{theo:cs}); \\
\Return $T_C$; \label{cbi-c:return}
\end{algorithm}

\begin{algorithm}[t]
\small
\caption{\small \textsc{CBI-Querying}}\label{alg:CBI-Querying}
\KwIn{a time-window: $[t_s, t_e]$; a CB-index: $T_{C}$;}
\KwOut{temporal butterfly counting: $num_{[t_s, t_e]}$;}
$num_{[t_s, t_e]} \leftarrow 0$; \\ \label{cbi-q:init}
\For{$(x, z) \in T_{C}.keys()$}{ \label{cbi-q:enumerate}
    $count \leftarrow T_{C}[(x,z)].query([t_s, \infty] \times [-\infty, t_e])$; \\ \label{cbi-q:query}
    $num_{[t_s, t_e]} \leftarrow num_{[t_s, t_e]} + \binom{count}{2}$ \label{cbi-q:sum}
}
\Return $num_{[t_s, t_e]}$; \label{cbi-q:return}
\end{algorithm}

\subsection{Graph Structure-aware Index}

Recall that the number of butterflies dominates the memory usage of \texttt{EBI}, and the number of wedge groups greatly affects the query time of \texttt{CBI}. They fail to take into account the actual structure of the input graph.
If we assume that most wedge groups only contain a very small amount of wedges, we might be able to take advantage of both \texttt{EBI} and \texttt{CBI} by distributing each wedge group into one of $\texttt{EBI}$ and $\texttt{CBI}$ by its size. As the main result of our paper, we propose a new index algorithm named \textit{Graph Structure-aware Index (\texttt{GSI})} with such an idea. Besides its promising performance on real-world data, it can also be parallelized to reach an even better performance. The implementation details are provided in \Cref{alg:GSI-Construction} and \Cref{alg:GSI-Querying}.

\begin{figure}[t!]
    \centering
    \includegraphics[width=1\linewidth]{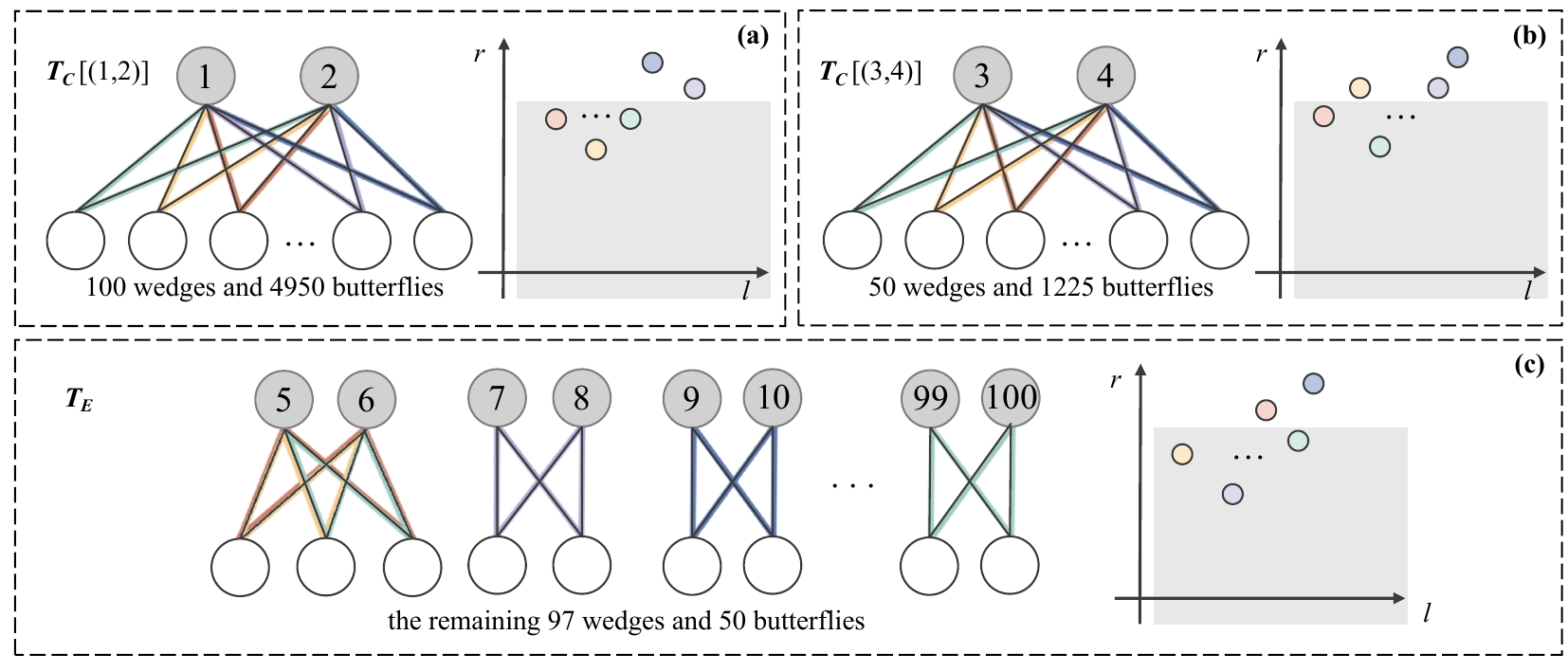}
    \caption{An illustrative example for demonstrating \texttt{GSI}.}
    \label{fig:GIS}
\end{figure}

$\bullet$ \textbf{Construction} The important assumption we make here is that a small number of groups contain most of the wedges. Our key idea is to store them in the same way as \texttt{CBI}, building a 2D-range counting data structure for each index, maintained by a hashmap $T_C$. For the rest of the wedges, we precompute all butterflies constructed from them and store them in one data structure $T_E$ like \texttt{EBI}.
In the beginning, in addition to $T_E$ (\Cref{gsi-c:initte}) and $T_C$ (\Cref{gsi-c:inittc}), we also initialize a hashmap $W$ to group wedges directly (\Cref{gsi-c:initw}).
We enumerate all $w_i : \left \langle x \leadsto y \leadsto z \right \rangle [l_i, r_i]$ in $G$ (\Cref{gsi-c:enumeratew}) and store them into the corresponding set mapped by $W[(x,z)]$ (\Cref{gsi-c:wupdate}) . If the key does not exist in $W$, we need to initialize the set to be empty (\Cref{gsi-c:neww}).
After grouping wedges, we sort the sets in $W$ with decreasing order of sizes. We also need to compute the total number of butterflies and store it in $num$ and $numB$ (\Cref{gsi-c:setnum} to \Cref{gsi-c:setnumb}).

We introduce a parameter $\alpha$ to control the cutoff between the two methods. More specifically, we build a data structure for some sets in $W$ such that the total number of butterflies constructed from them does not exceed $\alpha$ fraction of all the butterflies in $G$. The butterflies constructed from the other sets in $W$ are maintained by one data structure per set.

We enumerate the sets in $W$ by decreasing the order of sizes (\Cref{gsi-c:enumeratewkey}). $num$ represents the number of butterflies we processed. If we have not processed the majority of butterflies, i.e., when $num$ is no less than $\alpha \cdot numB$ (\Cref{gsi-c:ifbig}), we build a $\mathcal{CS}$ for every wedge in $W[(x,z)]$, inserting points representing their active timestamps (\Cref{gsi-c:newtc} to \Cref{gsi-c:inserttc}). $num$ is subtracted by $\binom{|W[(x, z)]|}{2}$, indicating that these of butterflies are recorded by $T_C[(x,z)]$ (\Cref{gsi-c:numupdate}).

When $num$ is smaller than $\alpha \cdot numB$, we have already handled the majority of butterflies. Now, we need to process butterflies built from the rest of the wedges groups. Though each of these groups only contains a relatively small amount of wedges, the number of groups might be large. Here we apply the idea of \texttt{EBI}. We enumerate every possible wedge pair in the current enumerated $W[(x,z)]$ (\Cref{gsi-c:forwi} to \Cref{gsi-c:forwj}), compute the butterfly $b_{ij}[t_l,t_r] \leftarrow w_i \circ w_j$ (\Cref{gsi-c:butterfly}) and insert a point $(t_l,t_r)$ into the global data structure $T_E$ (\Cref{gsi-c:insertte}).
In the end, after preprocessing every $\mathcal{CS}$ in $T_C$ and $T_E$, we return them and finish the construction (\Cref{gsi-c:return}).

An example of construction is shown in \Cref{fig:GIS}. In the input graph, there are 100 wedges with endpoints $(1,2)$ and 50 wedges with endpoints $(3,4)$, while there are only 97 wedges in other groups. We treat these two groups with the \texttt{CBI} method and build $T_C[(1,2)]$ (a), $T_C[(3,4)]$ (b) correspondingly. We compute all 50 butterflies for the remaining wedges and store them in $T_E$ (c).

$\bullet$ \textbf{Answering Query} For a historical butterfly counting of time-window $ [t_s, t_e] $, we need to accumulate answers from both $ T_C$ and $ T_E$. As in \texttt{EBI}, we directly query the 2D range $ [t_s, \infty] \times [-\infty, t_e] $ on $ T_E$ (\Cref{gsi-q:queryte}). Then, as in \texttt{CBI}, we enumerate all groups of wedges from $T_C$ (\Cref{gsi-q:enumerate}). We query $T_C[(x, z)]$ with the same 2D range (\Cref{gsi-q:querytc}) and add up the number of butterflies constructed from each group of wedges (\Cref{gsi-q:sum}). We return the answer as the sum of these two parts (\Cref{gsi-q:return}). In the \Cref{fig:GIS} showcase, we query from all three $\mathcal{CS}$s: $T_C[(1,2)]$, $T_C[(3,4)]$, $T_E$ and sum up the answer. 

\label{sec:gsi}
\begin{algorithm}[t]
\small
\caption{\small \textsc{GSI-Construction}}\label{alg:GSI-Construction}
\KwIn{a bipartite temporal graph: $G = (V, E)$; a parameter for space using: $\alpha$; }
\KwOut{a $\mathcal{CS}$ for indexing (EB-index): $T_{E}$; a hashmap of $\mathcal{CS}$ for indexing: $T_{C}$;}

Initialize $T_{E}$ with an empty $\mathcal{CS}$; \\\label{gsi-c:initte}
$T_C \leftarrow$ an empty hashmap mapping pairs of vertices to $\mathcal{CS}$s; \\ \label{gsi-c:inittc}
$W \leftarrow$ an empty hashmap mapping pairs of vertices to sets; \\ \label{gsi-c:initw}
\For{each wedge $w_i : \left \langle x \leadsto y \leadsto z \right \rangle [l_i, r_i] \in G$}{ \label{gsi-c:enumeratew}
    \lIf{$(x, z) \notin W.keys()$}{ \label{gsi-c:if}
        $W[(x, z)] \leftarrow \emptyset$ \label{gsi-c:neww}
    }
    $W[(x, z)] \leftarrow W[(x, z)] \cup w_i$; \\ \label{gsi-c:wupdate}
}
Sort sets in $W$ with decreasing order of sizes; \\ \label{gsi-c:sort}
$num \leftarrow 0$; \\ \label{gsi-c:setnum}
\lFor{$(x, z) \in W.keys()$}{ \label{gsi-c:enumeratekey}
    $num \leftarrow num + \binom{|W[(x, z)]|}{2}$ \label{gsi-c:computebutterfly}
}
$numB \leftarrow num$; \\ \label{gsi-c:setnumb}
\For{$(x, z) \in W.keys()$}{ \label{gsi-c:enumeratewkey}
    \If{$num \geq \alpha \cdot numB$}{ \label{gsi-c:ifbig}
        Initialize $T_{C}[(x, z)]$ with an empty $\mathcal{CS}$; \\ \label{gsi-c:newtc}
        \For{$w_i : [l_i, r_i] \in W[(x, z)]$}{ \label{gsi-c:forwedge}
            $T_{C}[(x, z)].insert(l_i, r_i)$; \label{gsi-c:inserttc}
        }
    }
    \Else{
        \For{$w_i : [l_i, r_i] \in W[(x, z)]$}{ \label{gsi-c:forwi}
            \For{$w_j : [l_j, r_j] \in W[(x, z)] \land (j > i)$} { \label{gsi-c:forwj}
                $b_{ij} \leftarrow w_i \circ w_j$; \\ \label{gsi-c:butterfly}
                $T_E.insert(t_l, t_r)$\tcp*{$(t_{l}, t_{r})=\mathcal{T}(b_{ij})$} \label{gsi-c:insertte}
            }
        }
    }
    $num \leftarrow num - \binom{|W[(x, z)]|}{2}$; \\ \label{gsi-c:numupdate}
}
Preprocess $T_E$ and each data structure in $T_C$ (\Cref{theo:cs}); \\
\Return $T_E, T_C$; \label{gsi-c:return}
\end{algorithm}

\begin{algorithm}[t]
\small
\caption{\small \textsc{GSI-Querying}}\label{alg:GSI-Querying}
\KwIn{a time-window: $[t_s, t_e]$; GSI-indexes: $T_{E}, T_{C}$;}
\KwOut{temporal butterfly counting: $num_{[t_s, t_e]}$;}
$num_{[t_s, t_e]} \leftarrow T_{E}.query([t_s, \infty] \times [-\infty, t_e])$; \\ \label{gsi-q:queryte}
\For{$(x, z) \in T_{C}.keys()$}{ \label{gsi-q:enumerate}
    $count \leftarrow T_{C}[(x,z)].query([t_s, \infty] \times [-\infty, t_e])$; \\ \label{gsi-q:querytc}
    $num_{[t_s, t_e]} \leftarrow num_{[t_s, t_e]} + \binom{count}{2}$ \label{gsi-q:sum}
}
\Return $num_{[t_s, t_e]}$; \label{gsi-q:return}
\end{algorithm}

$\bullet$ \textbf{Time and Space Complexity} The time and space complexity of \texttt{GSI} is summarized as follows:
\begin{theorem}
\label{thm:complexity}
Denote $\beta$ as the number of $\mathcal{CS}$ building for wedge groups of the input graph $G$ and $\delta$ as the arboricity.
\texttt{GSI}'s construction (\Cref{alg:GSI-Construction}) runs in $O(m\delta\log n + \alpha B \log m)$ and takes $O(m\delta + \alpha B)$ space. After the construction, each query (\Cref{alg:GSI-Querying}) takes $O(\beta \log n + \log m)$ to answer.     
\end{theorem}
\begin{proof}
    The construction process (\Cref{alg:GSI-Construction}) requires enumerating all wedges, which takes $O(m\delta)$ by \cite{wang2019vertex}. The rest of the algorithm can be divided into two parts: a \texttt{CBI} instance of $\beta$ wedges group containing $(1-\alpha)B$ butterflies and an \texttt{EBI} instance containing $\alpha B$ butterflies. By \Cref{theo:EBI} and \Cref{theo:CBI}, the total time complexity is $O(m\delta\log n + \alpha B \log m)$ and the corresponding query complexity (\Cref{alg:GSI-Querying}) is $O(\beta \log n + \log m)$. The memory usage is $O(m\delta + \alpha B)$.
\end{proof}
Note that $\beta$ is determined by $\alpha$ and bounded by the number of wedge groups $\widetilde{w}$.
In real-world datasets with practical memory constraints, $\beta$ is significantly smaller than $\widetilde{w}$ (\textit{i.e.,} $\beta \ll \widetilde{w}$). 
For example, in the \texttt{WT} dataset and under a 500 GB memory constraint, the $\beta$ of \texttt{GSI} is $6986$ when setting the optimal $\alpha$ that maximize the memory utilization, while the $\widetilde{w}$ is $3,925,064$.

As a result, \texttt{GSI} manages to address the concentration phenomenon of the wedges distribution on real-world graphs, taking advantage of both \texttt{EBI} (for scattered wedges) and \texttt{CBI} (for concentrated wedges). It also has flexibility regarding the different limitations of memory. In addition, under the well-known power-law model \cite{aiello2000random}, \texttt{GSI} has a non-trivial complexity, which is analyzed in \Cref{sec:power-law}.

\subsubsection{Automatically Determining $\alpha$}\label{sec:alpha}
The parameter $\alpha$ reflects a trade-off between query time and memory usage. Therefore, selecting a proper $\alpha$ based on the input graph and the memory limit is important. We introduce an easy method to automatically select an efficient $\alpha$ before executing \texttt{GSI}. In the beginning, we set $\alpha=0$, which means that all wedge groups will be maintained separately with a $\mathcal{CS}$ data structure. We sort wedge groups by size from small to large and enumerate them one by one. For each enumerated wedge group $i$ with $w_i$ wedges, we increase $\alpha$ by $\binom{w_i}{2}/B$, where $B$ is the number of butterflies, precomputed in \Cref{gsi-c:setnum} to \Cref{gsi-c:computebutterfly}. We can estimate the memory usage by simple calculation as Chazelle’s structure's memory usage can be precomputed based on the number of inserted points in $O(1)$. We stop and determine $\alpha$ if the estimated memory exceeds the limit with the next wedge group. This process takes $O(m\delta)$ time and avoids actually building indexes.

\subsubsection{Parallelized Querying}
After the construction process, when answering a query (\Cref{alg:GSI-Querying}), we are, in fact, summing up the answer from several independent $\mathcal{CS}$ data structures, which motivates us to utilize parallelization to speed up.
Assume we have multiple threads, and these threads can handle different $\mathcal{CS}s$ simultaneously. 
Since no conflict occurs when these threads read the indexes simultaneously, we consider \Cref{alg:GSI-Querying} is highly parallelizable when the number of $\mathcal{CS}$s greater than the number of threads.
To demonstrate, we implement and evaluate a parallel version of \texttt{GSI} in \Cref{sec:eval-parallel}.

\subsection{Index Compression}
\label{sec:compression}
\texttt{EBI} has already reached the memory lower-bound in \Cref{thm:hardness}. Meanwhile, \texttt{GSI} provides a lower and adjustable memory usage in practice, with a better theoretical guarantee in the power-law graph.
Both our algorithms are designed for exact solutions. 
However, in some applications of butterfly counting, such as graph kernel analysis \cite{sheshbolouki2022sgrapp} and network measurement~\cite{sanei2018butterfly, wang2014rectangle}, using approximation counts is sufficient.
In this section, we show that if we allow a small error ratio (e.g., $\leq 10^{-6}$\% ) of the butterfly counting, we can reach a much lower memory usage on real-world graphs. Two index compression methods for approximate historical butterfly counting with strong theoretical guarantees are introduced.

\paragraph{Single-sided Compression (SGSI)}
Intuitively, most butterflies in \Cref{alg:GSI-Querying} come from $T_C$ (\Cref{gsi-q:enumerate} to \Cref{gsi-q:querytc}) while the actual memory usage of $T_E$ is much larger in general.
This inspires us only to maintain a small portion (e.g., $\frac{1}{1000}$) of all $T_E$'s butterflies.
We denote $\lambda_1$ as the single-sided compression ratio,  which indicates we only preserve $\lambda_1 \cdot \alpha B$ butterflies, where $\alpha B$ is the total number of butterflies in $T_E$. Each butterfly is preserved with probability $\lambda_1$ independently. In other words, we modify the \texttt{GSI} algorithm so that when \texttt{GSI} inserts a 2D point into $T_E,$ the insertion is actually performed with probability $\lambda_1.$ 
We denote the compressed $T_E$ by $\widetilde{T_E}.$ To answer queries, we return $\widetilde{T_E}.query([t_s, \infty]\times [-\infty, t_e]) / \lambda_1$ (\Cref{gsi-q:queryte} in \Cref{alg:GSI-Querying}).
This results in a compressed memory usage of $O(m\delta\log n + \lambda_1 \cdot \alpha B \log m)$.

\paragraph{Double-sided Compression (DGSI)} \texttt{SGSI} only compresses on $T_E$, which does not affect the query efficiency. If we further compress $T_C$, we can reach a faster query time and even lower memory usage. Specifically, in \texttt{DGSI}, besides compressing $T_E$ as in \texttt{SGSI}, we only maintain $\lambda_2$ portion of wedges for each group in $T_C$'s $\mathcal{CS}$s. Each wedge is preserved with probability $\lambda_2$ independently.
We denote the compressed $T_C$ by $\widetilde{T_C}.$ To answer queries, each $\widetilde{\mathcal{CS}}$ in $\widetilde{T_C}$ contributes $\binom{count}{2}/\lambda_2^2$ (\Cref{gsi-q:sum} in \Cref{alg:GSI-Querying}) to the total answer where $count$ is the number of wedges in $\widetilde{\mathcal{CS}}$ in the projected graph.
This results in an $O(\beta \log (\lambda_2 n) + \log (\lambda_1 m))$ query time, which is more efficient than our exact algorithm, \texttt{GSI}.

In our technical report Part \ref{append:compression}, we prove that both compression methods are unbiased and bound their errors by the following theorems.
\begin{theorem}[Compressing $T_E$]
\label{thm:SGSI}
    For any historical butterfly counting query $[t_s, t_e]$ on a bipartite graph $G$, let $B$ denote the correct number of butterflies in $G_{[t_s, t_e]}$ that is maintained by $T_E,$ i.e., $B=T_E.query([t_s, \infty]\times [-\infty, t_e]).$ Let $B'$ denote the number of butterflies computed by the compressed $T_E$ in \texttt{SGSI} with compression ratio $\lambda_1$, i.e., $B'=\widetilde{T_E}.query([t_s, \infty]\times [-\infty, t_e]) / \lambda_1.$
    $B'$ is unbiased, i.e., the expectation of $B'$ is equal to $B.$ The absolute error $|B-B'|$ is $O({B}^{0.5}\log^{0.5}(n)\lambda_1^{-0.5})$ with high probability, where $n$ is the number of vertices.
\end{theorem}

\begin{theorem}[Compressing $T_C$]
\label{thm:DGSI}
    For any historical butterfly counting query $[t_s, t_e]$ on a bipartite graph $G$, let $S$ denote the correct number of butterflies in $G_{[t_s, t_e]}$ that is maintained by $T_C,$ i.e., the total increment of $num_{[t_s, t_e]}$ in \Cref{gsi-q:sum} of \Cref{alg:GSI-Querying}. Let $S'$ denote the number of butterflies computed by the compressed $\widetilde{T_C}$ in \texttt{DGSI} with compression ratio $\lambda_2$. $S'$ is unbiased, i.e., the expectation of $S'$ is equal to $S.$ The absolute error $|S-S'|$ is $O({S}^{0.75}\ln(n)^{0.75}\lambda_2^{-1.75})$ with high probability, where $n$ is the number of vertices.
\end{theorem}

\Cref{thm:SGSI} directly bounds the error of \texttt{SGSI}. To bound the error of \texttt{DGSI}, we may add the error bounds for compressing the $T_E$ part (\Cref{thm:SGSI}) and compressing the $T_C$ part (\Cref{thm:DGSI}).

\section{ANALYSIS ON POWER-LAW GRAPHS}
\label{sec:power-law}

\hspace{1em} $\bullet$ \textbf{Power-law Bipartite Graphs} Previous research~\cite {guillaume2006bipartite, vasques2018degree} shows that many real-world bipartite graphs follow the \textit{power-law} distribution with $\gamma$ generally between $2$ to $3$, similar to the \textit{scale-free} graphs model for unipartite graphs.
By leveraging the properties of the power-law degree distribution, there are many studies on analyzing algorithms based on such settings~\cite{zhang2023scalable, yang2021efficient, latapy2008main, brach2016algorithmic}. 
In this section, we prove that \texttt{GSI} runs in $\Otil(\lambda)$ time and $O(m^{2-\delta}/\lambda^2)$ memory for some $\delta>0$ with high probability on power-law bipartite graphs~\cite{aiello2000random} with $\gamma \in (2, 3)$, in contrast to the hardness result in \Cref{sec:hardness} for the general case. 
We use the following model for $G$.

\begin{mydef}
\label{def:power-law}
Let $n_1$ be the number of vertices in $U$. Let $n_2$ be the number of vertices in $L$. Let $\gamma_i$, $\Delta_i$ ($i=1,2$) be the exponents and max degrees of two power-law distributions. A power-law bipartite graph is obtained by the following process: \textbf{(1)} For each vertex $x$ in $U$, sample $d_x\in [1,\Delta_1]$ such that $\pr{}{d_x=k}\propto k^{-\gamma_1}$. \textbf{(2)} For each vertex $y$ in $L$, sample $d_y\in [1, \Delta_2]$ such that $\pr{}{d_y=k}\propto k^{-\gamma_2}$. \textbf{(3)} Let $m=\expec{}{\sum_{x\in U}d_x}=\expec{}{\sum_{y\in L}d_y}$ be the expected number of edges in $G$. \textbf{(4)} For each pair of vertices $x \in U$, $y\in L$, sample the existence of the edge $(x, y)$ such that $\pr{}{(x, y)\text{exists}}=\frac{d_xd_y}{m}$.
\end{mydef}
For \Cref{def:power-law} to be well-defined, we need to choose parameters such that $\expec{}{\sum_{x\in U}d_x}= \expec{}{\sum_{y\in L}d_y}$, i.e., the expected sums of degrees for vertices in $U$ and $L$ are equal.
We note that the sampled $d_x$ and $d_y$ values are intermediate variables of the sampling process. They are not necessarily the same as the degrees $\deg_x$ and $\deg_y$ of vertex $x$ and $y$. We can interpret $d_x$ as the expected degree of $x$.
Based on \Cref{def:power-law}, there are two following two types of power-law bipartite graphs to model the real-world networks comprehensively.

\begin{mydef} [Double-sided Power-law Bipartite Graph]
\label{def:double-power-law}
    A double-sided power-law bipartite graph is a power-law bipartite graph (\Cref{def:power-law}) satisfying $\gamma_1\in (2, 3)$, and $\gamma_2\in (2, 3).$
\end{mydef}

\begin{mydef} [Single-sided Power-law Bipartite Graph]
\label{def:single-power-law}
    A single-sided power-law bipartite graph is a power-law bipartite graph (\Cref{def:power-law}) satisfying $\gamma_1\in (2, 3)$, $\gamma_2=0,$ and $\Delta_1 > \Delta_2.$
\end{mydef}

Twitter exemplifies a single-sided power-law bipartite graph, where the distribution of followers per user varies significantly, often displaying a vast disparity. In contrast, the number of accounts each user follows tends to be more uniform and comparatively narrow in range. 
In contrast, movie-actor networks exhibit a double-sided power-law distribution. Most actors appear in just a few films, and the majority of films feature only a small number of actors. However, a select group of highly connected actors mirrors the highly-followed users on social media platforms, while films with large casts resemble ``super-connected'' nodes.

$\bullet$ \textbf{Time and Space Complexity} We calculate the expected time and space of \texttt{GSI} for each type of power-law bipartite graph.
\begin{theorem}[\GSI~on double-sided power-law bipartite\\ graphs]
\label{thm:double_power_law}
    Let $G$ be a single-sided power-law bipartite graph.
    We can set $\alpha$ in \Cref{alg:GSI-Construction}, such that the expected space usage of \GSI~is $O(\Delta^{6-2\gamma})$, and that the query time is $\Otil(1).$ We can also set $\alpha$ such that the expected space usage is $O\left(n\Delta^{3-\gamma}\right)$, and that the expected query time is $\Otil(n^2\min(\Delta_1, \Delta_2)^{2-2\gamma}).$
\end{theorem}
We note that for both settings of \Cref{thm:double_power_law}, the time and space complexities surpass the lower bound in \Cref{thm:hardness}. This shows that the historical butterfly counting problem on power-law graphs are not as hard as the general case, and that \GSI~successfully exploits the features of the power-law graphs to reduce the computation cost. The first setting has $\Otil(1)$ time complexity and $o(n^2)$ space complexity since $\Delta\le n$ and $\gamma > 2.$ The second setting has $\Otil(\lambda)=\Otil(n^2\min(\Delta_1, \Delta_2)^{2-2\gamma})$ time complexity and $o(n^2/\lambda^2)$ space complexity when $k>n^{\frac{3}{4\gamma-4}}\Delta^{\frac{3-\gamma}{4\gamma-4}}.$

\begin{theorem}[\GSI~on single-sided power-law bipartite\\ graphs]
\label{thm:single_power_law}
    Let $G$ be a single-sided power-law bipartite graph.
    
    We can set $\alpha$ in \Cref{alg:GSI-Construction}, such that the expected space usage of \GSI~is $O\left(\Delta_2+\left(\frac{\Delta_1^{6-\gamma_1}}{\Delta_2^5}\right)\right)$ and that the expected query time is $\Otil\left(1\right).$
\end{theorem}

Similar to the double-sided case, \Cref{thm:single_power_law} solves the historical butterfly counting problem with a better trade-off than \Cref{thm:hardness}.

\modify{The complete proofs of \Cref{thm:double_power_law} and \Cref{thm:single_power_law} are provided in the technical report.}

\section{EXPERIMENTS}~\label{sec:exp}
\begin{table}
\centering
\caption{The summary of datasets, including the size of vertices (\textit{i.e.,} $|U|$ and $|L|$) and edges (\textit{i.e.,} $|E|$), butterfly counts (\textit{i.e.,} $B$), and the $\gamma$ value of power-law exponent assuming the graph's degree follows a power-law distribution (\textit{i.e.,} $\hat{\gamma}$). For synthetic datasets, we provide the $\gamma$ value for each side (\textit{i.e.,} $\gamma_{1}$ and $\gamma_{2}$).}

\label{tab:dataset}
\resizebox{1.0\linewidth}{!}{
\begin{tabular}{|c|rrrrc|} 
\hline
Real-world dataset & \multicolumn{1}{c}{$|U|$} & \multicolumn{1}{c}{$|L|$} & \multicolumn{1}{c}{$|E|$} & \multicolumn{1}{c}{$B$}          & $\hat{\gamma}$   \\ 
\hline
\hline
Wikiquote (WQ)     & 961                       & 640,482                   & 776,458                   & 5e9                              & 4.281            \\
Wikinews (WN)      & 2,200                     & 35,979                    & 907,499                   & 3e11                             & 2.654            \\
StackOverflow (SO) & 545,196                   & 96,680                    & 1,301,942                 & 1e7                              & 2.797            \\
CiteULike (CU)     & 153,277                   & 731,769                   & 2,411,819                 & 6e8                              & 2.325            \\
Bibsonomy (BS)     & 204,673                   & 767,447                   & 2,555,080                 & 8e8                              & 2.209            \\
Twitter (TW)       & 175,214                   & 530,418                   & 4,664,605                 & 1e12                             & 2.460            \\
Amazon (AM)        & 2,146,057                 & 1,230,915                 & 5,838,041                 & 3e7                              & 2.957            \\
Edit-ru (ER)       & 7,816                     & 1,266,349                 & 8,349,235                 & 1e13                             & 5.301  \\
Edit-vi (EV)       & 72,931                    & 3,512,721                 & 25,286,492                & 3e14                             & 2.017            \\
Wiktionary (WT)    & 66,140                    & 5,826,113                 & 44,788,448                & 1e16                             & 1.826            \\ 
\hline
\hline
Synthetic dataset & \multicolumn{1}{c}{$|U|$} & \multicolumn{1}{c}{$|L|$}   & \multicolumn{1}{c}{$|E|$} & \multicolumn{1}{c}{$\gamma_{1}$} & $\gamma_{2}$     \\ 
\hline
\hline
DoublePower-law (DPW)               & 25,599,878                     & 25,599,878                     & 100,000,000                       & \multicolumn{1}{c}{2.1}          & 2.1              \\

SinglePower-law (SPW)               & 21,020,985                       & 199,802                       & 100,000,000                       & \multicolumn{1}{c}{2.1}          & 0                \\
\hline
\end{tabular}
}
\end{table}

$\bullet$ \textbf{Datasets} We use 10 large-scale real-world datasets to evaluate our algorithm.
All real-world dataset sources and more detailed statistics are available at KONECT\footnote{from http://konect.cc/, where the $\hat{\gamma}$ of \texttt{ER} only considers the tail degrees (i.e., $\gamma_t$) because the global value missed in the website.}. 
In addition, we randomly generate 2 power-law bipartite graphs with billion edges according to the models in \Cref{sec:power-law}, where their timestamps are also uniformly sampled from $[1, |E|]$.
The dataset statistics are presented in \Cref{tab:dataset}.

$\bullet$ \textbf{Algorithms} Our algorithms are implemented in C++ with STL used.
The implementation of \texttt{BFC-VP++} is obtained from their authors.
As \texttt{WPS} is an approximation method, the reported running time for \texttt{WPS} represents the time required first to reach a relative error of 10\%.

$\bullet$ \textbf{Hyperparameter Settings}
The only hyperparameter for our methods is $\alpha$. In our study of space-query trade-offs (\Cref{sec:sq-trade-offs}), we determine the appropriate $\alpha$ using the algorithm described in Section \ref{sec:alpha} to constrain memory usage. We employ the same algorithm for all other experiments to optimize $\alpha$ for maximal memory utilization.

All experiments are conducted on a Ubuntu 22.04 LTS workstation, equipped with an Intel(R) Xeon(R) Gold 6338R 2.0GHz CPU and 512GB of memory.

\subsection{Efficiency}
\label{sec:eval-eff}
\subsubsection{Query Processing}

Shown in \Cref{fig:main_time}, the efficiency of \texttt{GSI}, as well as the two baseline algorithms \texttt{BFC-VP++}, \texttt{WPS}, is compared on various datasets, with $5000$ random queries. 

On real-world datasets, \texttt{GSI} demonstrates a significant speedup ranging from $100 \times$ to $10000 \times$ comparing to \texttt{BFC-VP++}, $100 \times$ to $1000 \times$ comparing to \texttt{WPS}. In most datasets, \texttt{GSI} needs less than $1$ seconds to answer all $5000$ queries online, even when the dataset has around $4 \times 10^6$ vertices (\texttt{AM}). In general, the query time is proportional to the number of edges. For the largest dataset \texttt{WT}, which has around $6 \times 10^6$ vertices and $4.5 \times 10^8$ edges, all queries can be answered within $10^2$ seconds, which means less than $0.1$ seconds for each query in average. This shows the consistency and efficiency of our algorithms under very large real-world datasets. 

On the synthetic dataset, we generated, our algorithm shows an extremely large speedup greater than $100000 \times$. Especially on \texttt{SPW}, while both \texttt{BFC-VP++} and \texttt{WPS} cannot process all queries within $10^5$ seconds, our algorithms produce answers within $10^{-1}$ seconds.

\input{99plot.tex}

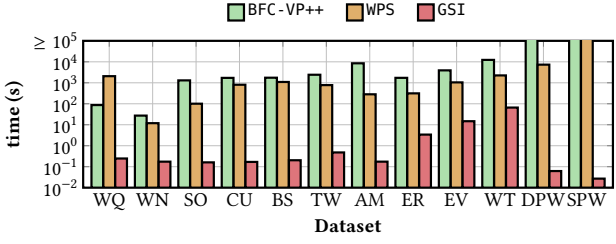
\begin{figure}[t!]
    \centering
       	\begin{tikzpicture}[scale=0.45]
        		\begin{axis}[
                    grid = major,
        			ybar=0.11pt,
        			bar width=0.3cm,
        			width=0.85\textwidth,
    				height=0.3\textwidth,
        			xlabel={\huge \bf Dataset}, 
        			xtick=data,	xticklabels={WQ, WN, SO, CU, BS, TW, AM, ER, EV, WT, DPW, SPW},
                     legend style={at={(0.5,1.30)}, anchor=north,legend columns=-1,draw=none},
                           legend image code/.code={
                    \draw [#1] (0cm,-0.263cm) rectangle (0.4cm,0.15cm); },
        			xmin=0.8,xmax=25.2,
    					ymin=0.01, ymax = 100000,
                         ytick = {0.01, 0.1, 1, 10, 100, 1000, 10000, 100000},
    	        yticklabels = {$10^{-2}$, $10^{-1}$, $10^0$,$10^1$, $10^2$, $10^3$, $10^4$, $\geq 10^{5}$},
                    ymode = log,    
                        log basis y={2},
                        log origin=infty,
        			tick align=inside,
        			ticklabel style={font=\huge},
        			every axis plot/.append style={line width = 1.6pt},
        			every axis/.append style={line width = 1.6pt},
                        ylabel={\textbf{\huge time (s)}},
        			]
        			\addplot[fill=p1] table[x=datasets,y=VP]{\HBFC};
        			\addplot[fill=p2] table[x=datasets,y=APPROX]{\HBFC};
           			\addplot[fill=p4] table[x=datasets,y=GIS]{\HBFC};
                \legend{\huge {\tt BFC-VP++ $\ $},\huge {\tt WPS $\ $}, \huge {\tt GSI}}
            		\end{axis}
        	\end{tikzpicture}
    	\caption{Time cost of 5,000 random queries for historical butterfly counting on datasets, where all algorithms use sequential versions.}
    \label{fig:main_time}
\end{figure}

\subsubsection{Index Construction}
The reported time for index construction is reported in \Cref{tab:index}.
For most datasets, even considering the time of index construction, \texttt{GSI}'s performance outperforms the two baselines, indicating that the index construction performance of \texttt{GSI} is both reasonable and practical.

\begin{table}
\centering
\caption{The index construction time for \texttt{GSI} in \Cref{fig:main_time}.}
\label{tab:index}
\resizebox{0.65\linewidth}{!}{
\label{tab:dataset}
\begin{tabular}{|c|cccccc|} 
\hline
Dataset  & WQ                     & WN    & SO    & CU    & BS    & TW   \\
time (s) & 74                     & 6     & 9     & 129   & 207   & 217  \\ 
\hline
\hline
Dataset  & AM & ER    & EV    & WT    & DPW   & SPW  \\ 
time (s) & 39                     & 3,654 & 3,842 & 3,964 & 2,566 & 341  \\
\hline
\end{tabular}
}
\end{table}

\pgfplotstableread[row sep=\\,col sep=&]{
	k &GIS & BFC \\
    0 & 0.047 & 2027.211 \\
	  2 & 0.053 & 1036.095 \\ 
        4 & 0.058 & 501.685  \\  
	6 & 0.063 & 269.639 \\ 
	8 & 0.105 & 177.810 \\ 
        10 & 0.207 & 138.831 \\ 
}\paraTW

\pgfplotstableread[row sep=\\,col sep=&]{
	k &GIS & BFC \\
    0 & 1.541 & 1566.944 \\
	  2 & 0.795 & 871.201 \\ 
        4 & 0.436 & 551.605  \\  
	6 & 0.263 & 322.960 \\ 
	8 & 0.256 & 271.642 \\ 
        10 & 0.610 & 251.374 \\ 
}\paraER

\pgfplotstableread[row sep=\\,col sep=&]{
	k &GIS & BFC \\
    0 & 11.265 & 4157.672 \\
	  2 & 5.123 & 2488.491 \\ 
        4 & 2.669 & 1546.447  \\  
	6 & 1.547 & 1076.379 \\ 
	8 & 0.872 & 885.177 \\ 
        10 & 0.633 & 756.897 \\ 
}\paraEV

\pgfplotstableread[row sep=\\,col sep=&]{
	k &GIS & BFC \\
    0 & 89.235 & 12570.343 \\
	  2 & 36.652 & 6792.164 \\ 
        4 & 20.057 & 3593.446  \\  
	6 & 11.812 & 2449.515 \\ 
	8 & 6.030 & 1843.104 \\ 
        10 & 4.973 & 1573.131 \\ 
}\paraWT

\begin{figure}[t!]
    \small
	\centering
    \subfigure[TW]{
	    \begin{tikzpicture}[scale=0.38]
			\begin{axis}[
              legend style = {
                legend columns=-1,
                draw=none,
            },
            width=0.53\textwidth,
            height=.36\textwidth,
            xtick = {0, 2, 4, 6, 8, 10},
            xticklabels={1, 2, 4, 8, 16, 32},
            ymode = log,
            mark size=6.0pt, 
            ymax = 12000,
            ylabel={\Huge \bf time (s)},
            xlabel={\Huge \bf number of threads}, 
            ticklabel style={font=\Huge},
            every axis plot/.append style={line width = 2.5pt},
            every axis/.append style={line width = 2.5pt},
            ]
            \addplot [mark=o,color=c2] table[x=k,y=GIS]{\paraTW};
            \addplot [mark=square,color=c8] table[x=k,y=BFC]{\paraTW};
            \legend{ \huge \texttt{GSI},  \huge \texttt{BFC-VP++}};
			\end{axis}
		\end{tikzpicture}
    }
    \subfigure[ER]{
	    \begin{tikzpicture}[scale=0.38]
			\begin{axis}[
              legend style = {
                legend columns=-1,
                draw=none,
            },
            width=0.53\textwidth,
            height=.36\textwidth,
            xtick = {0, 2, 4, 6, 8, 10},
            xticklabels={1, 2, 4, 8, 16, 32},
            ymode = log,
            mark size=6.0pt, 
            ymax = 12000,
            ylabel={\Huge \bf time (s)},
            xlabel={\Huge \bf number of threads}, 
            ticklabel style={font=\Huge},
            every axis plot/.append style={line width = 2.5pt},
            every axis/.append style={line width = 2.5pt},
            ]
            \addplot [mark=o,color=c2] table[x=k,y=GIS]{\paraER};
            \addplot [mark=square,color=c8] table[x=k,y=BFC]{\paraER};
            \legend{ \huge \texttt{GSI},  \huge \texttt{BFC-VP++}};
			\end{axis}
		\end{tikzpicture}
    }
    \subfigure[EV]{
	    \begin{tikzpicture}[scale=0.38]
			\begin{axis}[
              legend style = {
                legend columns=-1,
                draw=none,
            },
            width=0.53\textwidth,
            height=.36\textwidth,
            xtick = {0, 2, 4, 6, 8, 10},
            xticklabels={1, 2, 4, 8, 16, 32},
            ymode = log,
            mark size=6.0pt, 
            ymax = 100000,
            ylabel={\Huge \bf time (s)},
            xlabel={\Huge \bf number of threads}, 
            ticklabel style={font=\Huge},
            every axis plot/.append style={line width = 2.5pt},
            every axis/.append style={line width = 2.5pt},
            ]
            \addplot [mark=o,color=c2] table[x=k,y=GIS]{\paraEV};
            \addplot [mark=square,color=c8] table[x=k,y=BFC]{\paraEV};
            \legend{ \huge \texttt{GSI},  \huge \texttt{BFC-VP++}};
			\end{axis}
		\end{tikzpicture}
    }
    \subfigure[WT]{
	    \begin{tikzpicture}[scale=0.38]
			\begin{axis}[
              legend style = {
                legend columns=-1,
                draw=none,
            },
            width=0.53\textwidth,
            height=.36\textwidth,
            xtick = {0, 2, 4, 6, 8, 10},
            xticklabels={1, 2, 4, 8, 16, 32},
            ymode = log,
            mark size=6.0pt, 
            ymax = 100000,
            ylabel={\Huge \bf time (s)},
            xlabel={\Huge \bf number of threads}, 
            ticklabel style={font=\Huge},
            every axis plot/.append style={line width = 2.5pt},
            every axis/.append style={line width = 2.5pt},
            ]
            \addplot [mark=o,color=c2] table[x=k,y=GIS]{\paraWT};
            \addplot [mark=square,color=c8] table[x=k,y=BFC]{\paraWT};
            \legend{ \huge \texttt{GSI},  \huge \texttt{BFC-VP++}};
			\end{axis}
		\end{tikzpicture}
    }
	\caption{Evaluating the parallelization: time cost for 5,000 random queries against the number of threads.}
	\label{fig:parallel}
\end{figure}
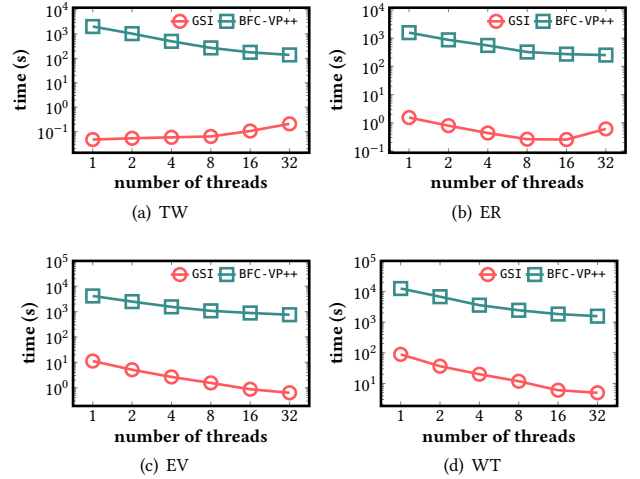

\subsubsection{Empirical Study on Power-law Graphs}
In this section, we further analyze the performance of \texttt{GSI} and baselines on the power-law graphs.
In addition to two synthetic graphs in \Cref{tab:dataset}, we generate power-law graphs with various $\gamma_1$ and $\gamma_2$ between $(2, 3)$ and uniformly sampled timestamps.
Specifically, to ensure that baselines can finish in practical time (\textit{i.e.,} less than $\leq 10^{5}$ seconds), we reduce the number of edges in these graphs to $10^7$ instead of $10^8$.
in \texttt{DPW} and \texttt{SPW}.
Shown in \Cref{fig:power-law-time}, both \texttt{BFC-VP++} and \texttt{WPS} require at least $10^3$ seconds in all synthetic data we generate while our algorithm can process all 5,000 queries within $10^{-1}$ seconds. This gap is even more significant compared to results on real-world data. Combined with the theoretical analysis in \cref{sec:power-law}, \texttt{GSI} has outstanding performance both theoretically and practically on Power-law graphs. This further reveals the advantage that \texttt{GSI} admits and makes full use of additional properties on graphs.


\begin{figure}[t!]
    \centering
       	\begin{tikzpicture}[scale=0.4]
        		\begin{axis}[
                    grid = major,
        			ybar=0.11pt,
        			bar width=0.45cm,
        			width=0.8\textwidth,
    				height=0.4\textwidth,
        			xlabel={\huge \bf Dataset}, 
        			xtick=data,	xticklabels={\Huge $\gamma_1=2.1\atop\gamma_2=2.1$, \Huge $\gamma_1=2.5\atop\gamma_2=2.5$, \Huge $\gamma_1=2.9\atop\gamma_2=2.9$, \Huge $\gamma_1=2.1\atop\gamma_2=0$, \Huge $\gamma_1=2.5\atop\gamma_2=0$, \Huge $\gamma_1=2.9\atop\gamma_2=0$},
                     legend style={at={(0.5,1.30)}, anchor=north,legend columns=-1,draw=none},
                           legend image code/.code={
                    \draw [#1] (0cm,-0.263cm) rectangle (0.4cm,0.15cm); },
        			xmin=0.8,xmax=13.2,
    					ymin=0.001, ymax = 100000,
                        ytick = {0.001, 0.01, 0.1, 1, 10, 100, 1000, 10000, 100000},
    	            yticklabels = {$10^{-3}$, $10^{-2}$, $10^{-1}$, $10^0$,$10^1$, $10^2$, $10^3$, $10^4$, $\geq 10^{5}$},
                        ymode = log,    
                        log basis y={2},
                        log origin=infty,
        			tick align=inside,
        			ticklabel style={font=\huge},
        			every axis plot/.append style={line width = 1.6pt},
        			every axis/.append style={line width = 1.6pt},
                        ylabel={\textbf{\huge time (s)}},
        			]
        			\addplot[fill=p1] table[x=datasets,y=VP]{\HBFCG};
        			\addplot[fill=p2] table[x=datasets,y=APPROX]{\HBFCG};
           			\addplot[fill=p4] table[x=datasets,y=GIS]{\HBFCG};
                \legend{\huge {\tt BFC-VP++ $\ $},\huge {\tt WPS $\ $}, \huge {\tt GSI}}
            		\end{axis}
        	\end{tikzpicture}
    	\caption{Time cost of 5,000 random queries for historical butterfly counting on different power-law bipartite graphs, where all algorithms use sequential versions.}
    \label{fig:power-law-time}
\end{figure}
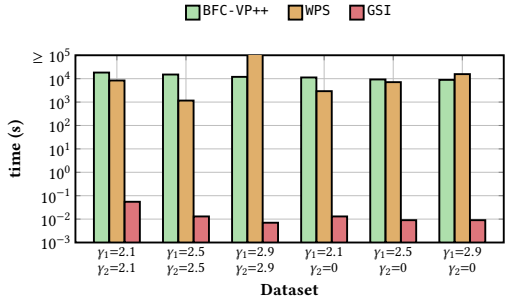

\subsubsection{Parallelization}
\label{sec:eval-parallel}
In addition, we compare \texttt{GSI}'s parallelized querying with \texttt{BFC-VP++}'s parallelized version.
In \Cref{fig:parallel}, we test the running time for 5,000 queries with various numbers of threads from $1$, $2$, $4$, $8$, $16$, and $32$.
In two large-scale datasets \texttt{EV} and \texttt{WT}, our algorithm obtains higher parallelism compared to the baseline, which has spent efforts in optimizing for higher parallelism.
On the other hand, in two other two large-scale datasets \texttt{TW} and \texttt{ER}, despite its fast runtime, our algorithm's performance does not always improve as the number of threads increases, which is not surprising and is under our expectations. This is because every historical query is essentially querying on multiple $\mathcal{CS}$s and summing up the answer. In the parallel version, we allocate threads to these data structures to increase efficiency. However, given a fixed memory limitation, \texttt{GSI} may create less $\mathcal{CS}$s if the number of motifs is relatively small. In such a circumstance, if the number of threads approaches or exceeds the number of data structures we built in construction, the runtime may increase. In general, our parallelized \texttt{GSI} still obtains a faster runtime and, most likely, better parallelism on very large-scale datasets.

\subsection{Comparison with TBC++}

\begin{figure}[t!]
    \centering
       	\begin{tikzpicture}[scale=0.4]
        		\begin{axis}[
                    grid = major,
        			ybar=0.11pt,
        			bar width=0.4cm,
        			width=0.85\textwidth,
    				height=0.3\textwidth,
        			xlabel={\huge \bf Dataset}, 
        			xtick=data,	xticklabels={WQ, WN, SO, CU, BS, TW, AM, ER, EV, WT},
                     legend style={at={(0.5,1.30)}, anchor=north,legend columns=-1,draw=none},
                           legend image code/.code={
                    \draw [#1] (0cm,-0.263cm) rectangle (0.4cm,0.15cm); },
        			xmin=0.8,xmax=21.2,
    					ymin=0.01, ymax = 100000,
                         ytick = {0.01, 0.1, 1, 10, 100, 1000, 10000, 100000},
    	        yticklabels = {$10^{-2}$, $10^{-1}$, $10^0$,$10^1$, $10^2$, $10^3$, $10^4$, $\geq 10^{5}$},
                    ymode = log,    
                        log basis y={2},
                        log origin=infty,
        			tick align=inside,
        			ticklabel style={font=\huge},
        			every axis plot/.append style={line width = 1.6pt},
        			every axis/.append style={line width = 1.6pt},
                        ylabel={\textbf{\huge time (s)}},
        			]
        			\addplot[fill=p3] table[x=datasets,y=TBC]{\TBC};
           			\addplot[fill=p4] table[x=datasets,y=GIS]{\TBC};
                        \legend{\huge {\tt TBC++ $\ $}, \huge {\tt GSI}}
            		\end{axis}
        	\end{tikzpicture}
    	\caption{\modify{Comparison with ~\cite{cai2024efficient} on the TBC task with random $\delta$, where the time of \texttt{GSI} includes both index time and query time.}}
    \label{fig:tbc_time}
\end{figure}

\modify{
Note that \texttt{TBC++}/\texttt{TBE+} ~\cite{cai2024efficient} only focus on temporal butterfly counting/listing such that the duration time is no larger than a given threshold, which is a related but orthogonal work that does not consider querying on arbitrary time windows.
We do not adapt \texttt{TBC++} and compare it in our problem because \textbf{(1): } running \texttt{TBC++} each time for the project graph of each query will lead to the algorithm degrading to \texttt{BFC-VP++}, which we have already included as a baseline strategy. \emph{e.g.,} \texttt{TBC++} needs 149 seconds for 5,000 random queries on \texttt{WQ}, where \texttt{BFC-VP++} only needs 86 seconds;  
\textbf{(2): } building index based on \texttt{TBE+} requires unaffordable memory storage since the number of listed temporal butterflies will be too huge when the threshold gets larger, which can be the maximum length of the queried time window in our problem. \emph{e.g.,}  \texttt{TBE+} can not run on \texttt{WT} due to the number of $\geq 10^{10}$ butterflies in the graph~\cite{cai2024efficient}.  
Therefore, to comprehensively compare with~\cite{cai2024efficient}, we instead adapt our algorithms to the temporal butterfly counting problem, which is the key problem~\cite{cai2024efficient} aiming to solve and can also be naturally generalized by \texttt{GSI}.
}

\modify{
The definition of the temporal butterfly is defined as follows:
}
\begin{mydef}[\modify{Temporal Butterfly~\cite{cai2024efficient}}]
\modify{
Given a duration threshold $\delta$, a temporal butterfly is a sequence of 4 temporal edges in chronological order $\left \langle e_1, e_2, e_3, e_4\right \rangle$, s.t., (1) $e_1.t < e_2.t < e_3.t < e_4.t$, (2) $e_4.t - e_1.t \leq \delta$,
and (3) the induced graph is a butterfly. 
}
\end{mydef}

\modify{
For each dataset, we randomly select a threshold value ($\delta$) and compare the total time taken to count the temporal butterflies for these thresholds between \texttt{GSI} and \texttt{TBC++}.
Specifically, the adaptation of \texttt{GSI} employs the same mechanism for combining the \#wedges and \#butterfly indices but utilizes a new underlying data structure, rather than Chazelle’s~\cite{chazelle1988functional}, tailored specifically to this problem.
}

\modify{
As shown in \Cref{fig:tbc_time}, \texttt{GSI} outperform \texttt{TBC++} up to $\frac{1}{100}\times$ in the five largest datasets (\texttt{TW}, \texttt{AM}, \texttt{ER}, \texttt{EV}, \texttt{WT}). In the other relatively small datasets, our runtime is close to \texttt{TBC++}. This is because our runtime includes the time to build the index, and we only query on one fixed $\delta$.
If we consider calculating TBC for varying $\delta$s, our index-based algorithms will significantly outperform \texttt{TBC++} even in those datasets.
For example, for 50 randomly selected $\delta$s on \texttt{WN}, \texttt{GSI} only needs 132 seconds in total, where \texttt{TBC++} needs more than 1,000 seconds.
}

\subsection{Space-Time Trade-Offs}
\label{sec:sq-trade-offs}
Since \texttt{GSI} is a flexible algorithm that strikes to balance the trade-off between space and time, we further study its performance under different memory limitations. Shown in \Cref{fig:space-query}, the query time can be efficiently reduced if provided more memory space: When the applicable memory is enlarged to $10 \times$, the query time is decreased to $\frac{1}{10} \times$ to $\frac{1}{100} \times$. Furthermore, it is shown that even when the memory space is very limited, \texttt{GSI} can still obtain a good query efficiency. For example, on \texttt{WT}, only given $40$ Gigabytes memory, both the index-construction time and query time are less than $600$ seconds, which still holds a significant advantage compared to our baseline in \cref{fig:main_time}. 
In all, this study further reveals \texttt{GSI}'s flexibility on the trade-off between space and time. In practice, it can fit different circumstances while providing high-efficiency performance.

\pgfplotstableread[row sep=\\,col sep=&]{
	k &index & query \\
    0 & 47.474 & 4.610 \\
	  2 & 81.226 & 0.678 \\ 
        4 & 119.068 & 0.181  \\  
	6 & 145.891 & 0.124 \\ 
	8 & 179.357 & 0.079 \\ 
        10 & 206.966 & 0.061 \\ 
        12 & 231.280 & 0.074 \\
        14 & 251.230 & 0.051 \\
        16 & 234.054 & 0.077 \\
        18 & 244.007 & 0.064 \\
}\STW

\pgfplotstableread[row sep=\\,col sep=&]{
	k &index & query \\
    0 & 400.897 & 3.449 \\
	  2 & 814.758 & 2.648 \\ 
        4 & 1181.392 & 2.300  \\  
	6 & 1585.274 & 1.912 \\ 
	8 & 1989.827 & 1.833 \\ 
        10 & 2378.886 & 2.751 \\ 
        12 & 2850.443 & 3.142 \\
        14 & 3208.945 & 2.595 \\
        16 & 3703.218 & 2.051 \\
        18 & 4211.026 & 2.785 \\
}\SER

\pgfplotstableread[row sep=\\,col sep=&]{
	k &index & query \\
    0 & 408.355 & 63.205 \\
	  2 & 778.353 & 44.558 \\ 
        4 & 1148.471 & 25.541  \\  
	6 & 1540.818 & 19.204 \\ 
	8 & 1952.256 & 14.456 \\ 
        10 & 2281.093 & 12.444 \\ 
        12 & 2658.085 & 13.290 \\
        14 & 3069.892 & 10.294 \\
        16 & 3649.160 & 13.866 \\
        18 & 4311.679 & 9.368 \\
}\SEV

\pgfplotstableread[row sep=\\,col sep=&]{
	k &index & query \\
    0 & 557.329 & 338.214 \\
	  2 & 912.957 & 203.449 \\ 
        4 & 1320.389 & 163.618  \\  
	6 & 1641.523 & 112.919 \\ 
	8 & 2055.427 & 97.067 \\ 
        10 & 2514.685 & 105.157 \\ 
        12 & 2873.004 & 84.452 \\
        14 & 3240.532 & 78.072 \\
        16 & 3873.681 & 68.346 \\
        18 & 4170.433 & 65.193 \\
}\SWT

\begin{figure}[t!]
    \small
	\centering
    \subfigure[TW]{
	    \begin{tikzpicture}[scale=0.38]
			\begin{axis}[
              legend style = {
                legend columns=-1,
                draw=none,
            },
            width=0.53\textwidth,
            height=.36\textwidth,
            xtick = {0, 2, 4, 6, 8, 10, 12, 14, 16, 18},
            xticklabels={3, 6, 9, 12, 15, 18, 21, 24, 27, 30},
            ymode = log,
            mark size=6.0pt, 
            ymax = 10000,
            ylabel={\Huge \bf time (s)},
            xlabel={\Huge \bf memory usage (GB)}, 
            ticklabel style={font=\Huge},
            every axis plot/.append style={line width = 2.5pt},
            every axis/.append style={line width = 2.5pt},
            ]
            \addplot [mark=x,color=c2] table[x=k,y=index]{\STW};
            \addplot [mark=diamond,color=c8] table[x=k,y=query]{\STW};
            \legend{ \huge index,  \huge query};
			\end{axis}
		\end{tikzpicture}
    }
    \subfigure[ER]{
	    \begin{tikzpicture}[scale=0.38]
			\begin{axis}[
              legend style = {
                legend columns=-1,
                draw=none,
            },
            width=0.53\textwidth,
            height=.36\textwidth,
            xtick = {0, 2, 4, 6, 8, 10, 12, 14, 16, 18},
            xticklabels={4, 8, 12, 16, 20, 24, 28, 32, 36, 40},
            ymode = log,
            mark size=6.0pt, 
            ymax = 100000,
            ylabel={\Huge \bf time (s)},
            xlabel={\Huge \bf memory usage (10 GB)}, 
            ticklabel style={font=\Huge},
            every axis plot/.append style={line width = 2.5pt},
            every axis/.append style={line width = 2.5pt},
            ]
            \addplot [mark=x,color=c2] table[x=k,y=index]{\SER};
            \addplot [mark=diamond,color=c8] table[x=k,y=query]{\SER};
            \legend{ \huge index,  \huge query};
			\end{axis}
		\end{tikzpicture}
    }
    \subfigure[EV]{
	    \begin{tikzpicture}[scale=0.38]
			\begin{axis}[
              legend style = {
                legend columns=-1,
                draw=none,
            },
            width=0.53\textwidth,
            height=.36\textwidth,
            xtick = {0, 2, 4, 6, 8, 10, 12, 14, 16, 18},
            xticklabels={4, 8, 12, 16, 20, 24, 28, 32, 36, 40},
            ymode = log,
            mark size=6.0pt, 
            ymax = 100000,
            ylabel={\Huge \bf time (s)},
            xlabel={\Huge \bf memory usage (10 GB)}, 
            ticklabel style={font=\Huge},
            every axis plot/.append style={line width = 2.5pt},
            every axis/.append style={line width = 2.5pt},
            ]
            \addplot [mark=x,color=c2] table[x=k,y=index]{\SEV};
            \addplot [mark=diamond,color=c8] table[x=k,y=query]{\SEV};
            \legend{ \huge index,  \huge query};
			\end{axis}
		\end{tikzpicture}
    }
    \subfigure[WT]{
	    \begin{tikzpicture}[scale=0.38]
			\begin{axis}[
              legend style = {
                legend columns=-1,
                draw=none,
            },
            width=0.53\textwidth,
            height=.36\textwidth,
            xtick = {0, 2, 4, 6, 8, 10, 12, 14, 16, 18},
            xticklabels={4, 8, 12, 16, 20, 24, 28, 32, 36, 40},
            ymode = log,
            mark size=6.0pt, 
            ymax = 100000,
            ylabel={\Huge \bf time (s)},
            xlabel={\Huge \bf memory usage (10 GB)}, 
            ticklabel style={font=\Huge},
            every axis plot/.append style={line width = 2.5pt},
            every axis/.append style={line width = 2.5pt},
            ]
            \addplot [mark=x,color=c2] table[x=k,y=index]{\SWT};
            \addplot [mark=diamond,color=c8] table[x=k,y=query]{\SWT};
            \legend{ \huge index,  \huge query};
			\end{axis}
		\end{tikzpicture}
    }
	\caption{Space-query trade-offs study}
	\label{fig:space-query}
\end{figure}

\pgfplotstableread[row sep=\\,col sep=&]{
	k & space & accuracy & index & query\\
    0 & 512000 & 0 & 3842 & 14.742\\
            2 & 61430 & 0.00259002 & 859.143 & 5.371  \\  
	  4 & 10078 & 0.00547853 & 386.849 & 7.170 \\ 
	6 & 4984 & 0.0189161 & 313.082 & 6.148 \\ 
	8 & 532 & 0.0735832 & 280.107 & 3.411 \\ 
}\CompressEV

\pgfplotstableread[row sep=\\,col sep=&]{
	k & space & accuracy & index & query\\
    0 & 512000 & 0 & 3842 & 14.742\\
            2 & 62358 & 0.000000770473 & 859.143 & 5.371  \\  
	  4 & 11084 & 0.00000512308 & 386.849 & 7.170 \\ 
	6 & 5487 & 0.0146369 & 313.082 & 6.148 \\ 
	8 & 449 & 0.0963626 & 280.107 & 3.411 \\ 
}\CompressWT

\begin{figure}[t!]
    \centering
       \subfigure[EV]{
	    \begin{tikzpicture}[scale=0.3]
                \begin{axis}[
                    grid = major,
        			ybar=0.11pt,
        			bar width=0.45cm,
        			width=0.63\textwidth,
    				height=0.4\textwidth,
        			xlabel={\huge \bf Compression Ratios}, 
        			xtick=data,	xticklabels={\Huge $\lambda_{1}=100\% \atop \lambda_{2}=100\%$, \Huge $\lambda_{1}=8\% \atop \lambda_{2}=100\%$, \Huge $\lambda_{1}=1.8\% \atop \lambda_{2}=100\%$, \Huge $\lambda_{1}=1\% \atop \lambda_{2}=50\%$, \Huge $\lambda_{1}=0.1\% \atop \lambda_{2}=5\%$},,
                     legend style={at={(0.4,1.30)}, anchor=north,legend columns=-1,draw=none},
                           legend image code/.code={
                    \draw [#1] (0cm,-0.263cm) rectangle (0.6cm,0.35cm); },
        			xmin=-1,xmax=9,
    					ymin=10, ymax = 1000000,
                         ytick = {10, 100, 1000, 10000, 100000, 1000000},
    	        yticklabels = {$10^1$, $10^2$, $10^3$, $10^4$, $10^{5}$, $10^{6}$},
                    ymode = log,    
                        log basis y={2},
                        log origin=infty,
        			tick align=inside,
        			ticklabel style={font=\huge},
        			every axis plot/.append style={line width = 1.6pt},
        			every axis/.append style={line width = 1.6pt},
                        ylabel={\textbf{\huge memory usage (MB)}},
        			]
        			\addplot[fill=p1] table[x=k,y=space]{\CompressEV};
                \legend{\huge {\tt Index Size $\ $}}
            		\end{axis}
                 \begin{axis}[
                        axis y line*=right, 
                        hide x axis,
        			width=0.63\textwidth,
    				height=0.4\textwidth,
        			xlabel={\huge \bf Dataset}, 
        			xtick=data,	xticklabels={\Huge $\lambda_{1}=100\% \atop \lambda_{2}=100\%$, \Huge $\lambda_{1}=12\% \atop \lambda_{2}=100\%$, \Huge $\lambda_{1}=1.8\% \atop \lambda_{2}=100\%$, \Huge $\lambda_{1}=1\% \atop \lambda_{2}=50\%$, \Huge $\lambda_{1}=0.1\% \atop \lambda_{2}=5\%$},
                       mark size=6.0pt, ,
                     legend style={at={(0.7,1.30)}, anchor=north,legend columns=-1,draw=none},
        			xmin=-1,xmax=9,
    					ymin=0, ymax = 10,
                         ytick = {0, 0.001, 0.01, 0.1, 1, 10},
    	        yticklabels = {$0\%$, $10^{-3}\%$, $10^{-2}\%$, $10^{-1}\%$, $1\%$, $10^1$},
                    ymode = log,    
                        log basis y={2},
                        log origin=infty,
        			tick align=inside,
        			ticklabel style={font=\huge},
                        every axis plot/.append style={line width = 2.5pt},
                        every axis/.append style={line width = 2.5pt},
                        ylabel={\textbf{\huge relative error}},
        			]
        			\addplot[mark=o,color=c2] table[x=k,y=accuracy]{\CompressEV};
                \legend{\huge {\tt Error $\ $}}
            		\end{axis}
		\end{tikzpicture}
    }
           \subfigure[WT]{
	    \begin{tikzpicture}[scale=0.3]
                \begin{axis}[
                    grid = major,
        			ybar=0.11pt,
        			bar width=0.45cm,
        			width=0.63\textwidth,
    				height=0.4\textwidth,
        			xlabel={\huge \bf Compression Ratios}, 
        			xtick=data,	xticklabels={\Huge $\lambda_{1}=100\% \atop \lambda_{2}=100\%$, \Huge $\lambda_{1}=8\% \atop \lambda_{2}=100\%$, \Huge $\lambda_{1}=1.8\% \atop \lambda_{2}=100\%$, \Huge $\lambda_{1}=1\% \atop \lambda_{2}=50\%$, \Huge $\lambda_{1}=0.1\% \atop \lambda_{2}=5\%$},,
                     legend style={at={(0.4,1.30)}, anchor=north,legend columns=-1,draw=none},
                           legend image code/.code={
                    \draw [#1] (0cm,-0.263cm) rectangle (0.6cm,0.35cm); },
        			xmin=-1,xmax=9,
    					ymin=10, ymax = 1000000,
                         ytick = {10, 100, 1000, 10000, 100000, 1000000},
    	        yticklabels = {$10^1$, $10^2$, $10^3$, $10^4$, $10^{5}$, $10^{6}$},
                    ymode = log,    
                        log basis y={2},
                        log origin=infty,
        			tick align=inside,
        			ticklabel style={font=\huge},
        			every axis plot/.append style={line width = 1.6pt},
        			every axis/.append style={line width = 1.6pt},
                        ylabel={\textbf{\huge memory usage (MB)}},
        			]
        			\addplot[fill=p1] table[x=k,y=space]{\CompressWT};
                \legend{\huge {\tt Index Size $\ $}}
            		\end{axis}
                 \begin{axis}[
                        axis y line*=right, 
                        hide x axis,
        			width=0.63\textwidth,
    				height=0.4\textwidth,
        			xlabel={\huge \bf Dataset}, 
        			xtick=data,	xticklabels={\Huge $\lambda_{1}=100\% \atop \lambda_{2}=100\%$, \Huge $\lambda_{1}=12\% \atop \lambda_{2}=100\%$, \Huge $\lambda_{1}=1.8\% \atop \lambda_{2}=100\%$, \Huge $\lambda_{1}=1\% \atop \lambda_{2}=50\%$, \Huge $\lambda_{1}=0.1\% \atop \lambda_{2}=5\%$},
                       mark size=6.0pt, ,
                     legend style={at={(0.7,1.30)}, anchor=north,legend columns=-1,draw=none},
        			xmin=-1,xmax=9,
    					ymin=0, ymax = 10,
                         ytick = {0.000001, 0.00001, 0.0001, 0.001, 0.01, 0.1},
    	        yticklabels = {$10^{-6}\%$, $10^{-5}\%$, $10^{-4}\%$, $10^{-3}\%$, $10^{-2}\%$, $10^{-1}\%$},
                    ymode = log,    
                        log basis y={2},
                        log origin=infty,
        			tick align=inside,
        			ticklabel style={font=\huge},
                        every axis plot/.append style={line width = 2.5pt},
                        every axis/.append style={line width = 2.5pt},
                        ylabel={\textbf{\huge relative error}},
        			]
        			\addplot[mark=o,color=c2] table[x=k,y=accuracy]{\CompressWT};
                \legend{\huge {\tt Error $\ $}}
            		\end{axis}
		\end{tikzpicture}
    }
 
    	\caption{Evaluating memory usage and relative error in various compression ratio settings. Note that both our compression algorithms are not affect the query effciency.}
    \label{fig:index-compression}
\end{figure}

\subsection{Study on Compressed Indexes}
In addition, shown in \Cref{fig:index-compression}, we evaluate the index compression on two large-scale datasets \texttt{EV} and \texttt{WT}. With the single-sided compression algorithm \texttt{SGSI}, we are able to guarantee a relative error $\leq 10^{-3}$ while compressing memory usage to $\frac{1}{50}$ (i.e., ~500G $\rightarrow$ ~1G). With the double-sided compression algorithm \texttt{DGSI}, we achieve $\leq 10^{-1}$ relative error while compressing memory for around 1000 times (~500G $\rightarrow$ ~500MB).
This experiment validates the effectiveness of our \texttt{GSI} algorithm with index compression based on sampling. Moreover, it empirically provides a trade-off between memory and accuracy, i.e., if we need a higher accuracy guarantee, we use \texttt{SGSI}. To compress the memory further, we may use \texttt{DGSI} instead.
\subsection{Case Study}
\label{sec:casestudy}
To begin with, we acquired the author-publication datasets~\footnote{https://github.com/THUDM/citation-prediction}, which contains 10,419,221 author-publication records before 2016.
Based on that, we build an author-publication bipartite temporal graph containing 14,223,972 vertices and 10,419,216 edges.
Consider the projected graph  $G_{[t_s, t_e]}$, its corresponding BCC should be {\Large $\frac{4 \times \bowtie_{[t_s, t_e]}}{\ltimes_{[t_s, t_e]}}$}, where {\Large $\bowtie_{[t_s, t_e]}$} is the butterfly count in the projected graph $G_{[t_s, t_e]}$ and {\Large $\ltimes_{[t_s, t_e]}$} is the number of 3-paths in $G_{[t_s, t_e]}$.

In this study, we investigate the trend of the BCC over various two-year time windows.
To achieve that, we need to efficiently query several historical butterfly counts by \texttt{GSI}.
The result is shown in \Cref{fig:case-study-1} (a), where we set the length of the time windows as two years (e.g., 1995 - 1996).
Surprisingly, the value of BCCs doesn't consistently rise alongside the research community's overall trend toward collaboration. This implies that the research community might not always achieve greater cohesiveness through collaboration.
In contrast, in \Cref{fig:case-study-1} (a), BCCs exhibit an initial general increase (1985 - 2000) followed by a decrease (2000 - 2015).
We delve deeper into the underlying reasons for these trends and analyze the changing pattern of \textbf{(1)} \textit{the average number of publications per author within each time-window} and \textbf{(2)} \textit{the average number of unique collaborators per author within each time-window}, shown in \Cref{fig:case-study-1} (b).
As the average number of publications generally rises from 1985 to 2015, it suggests a higher probability of collaboration, consequently increasing the butterfly counts in the author-publication bipartite temporal graph.
Nevertheless, a greater likelihood of collaboration does not necessarily equate to more cohesive collaboration.
As depicted in \Cref{fig:case-study-1} (b), there is a notable increase in the average number of unique collaborators after 2000, indicating a rise in 3-paths counts within the graph, which may be an important factor causing BCCs decreases again from 2000 to 2010. Entering 2010, BCCs come across another rising trend. There are many possible reasons, one of which could be the significant increase in total publications.

\pgfplotstableread[row sep=\\,col sep=&]{
k & BCC & auth & cobo \\
0 & 0.0945 & 1.7965 & 2.7495 \\
2 & 0.0932 & 1.8421 & 2.8485 \\
4 & 0.0860 & 1.8809 & 2.9277 \\
6 & 0.0859 & 1.9159 & 2.9743 \\
8 & 0.0902 & 1.9634 & 3.1056 \\
10 & 0.0949 & 2.0044 & 3.2204 \\
12 & 0.0947 & 2.0639 & 3.3460 \\
14 & 0.0923 & 2.1239 & 3.5013 \\
16 & 0.0918 & 2.1659 & 3.6037 \\
18 & 0.0908 & 2.1879 & 3.6835 \\
20 & 0.0954 & 2.2320 & 3.7966 \\
22 & 0.0982 & 2.3020 & 3.9591 \\
24 & 0.1016 & 2.3515 & 4.1192 \\
26 & 0.1014 & 2.3872 & 4.2136 \\
28 & 0.1015 & 2.4356 & 4.2832 \\
30 & 0.0998 & 2.5092 & 4.5052 \\
32 & 0.0992 & 2.5909 & 4.7845 \\
34 & 0.0990 & 2.6641 & 5.0068 \\
36 & 0.0945 & 2.7185 & 5.1817 \\
38 & 0.0925 & 2.7782 & 5.3074 \\
40 & 0.0925 & 2.8296 & 5.3937 \\
42 & 0.0902 & 2.8727 & 5.5352 \\
44 & 0.0893 & 2.9164 & 5.6207 \\
46 & 0.0886 & 2.9600 & 5.6424 \\
48 & 0.0874 & 3.0143 & 5.7685 \\
50 & 0.0870 & 3.0761 & 5.8815 \\
52 & 0.0858 & 3.1384 & 6.0000 \\
54 & 0.0867 & 3.1804 & 5.9268 \\
56 & 0.0894 & 3.0407 & 6.0616 \\
58 & 0.0929 & 2.7777 & 6.4410 \\
}\CASE

\begin{figure}[t!]
    \centering
    \subfigure[BCC]{
    \begin{tikzpicture}[scale=0.38]
        \begin{axis}[
          legend style = {
            legend columns=-1,
            draw=none,
        },
        width=0.55\textwidth,
        height=0.36\textwidth,
        xtick = {0, 10, 20, 30, 40, 50, 60},
        xticklabels = {1985, 1990, 1995, 2000, 2005, 2010, 2015},
        mark size=2.0pt, 
        ymax = 0.105,
        ymin = 0.08,
        ytick = {0.07, 0.075, 0.08, 0.085, 0.09, 0.095, 0.10},
        yticklabels = {7, 7.5, 8, 8.5, 9, 9.5, 10},
        ylabel={\Huge \bf value ($10^{-2}$)},
        xlabel={\Huge \bf years}, 
        ticklabel style={font=\Huge},
        every axis plot/.append style={line width = 2.5pt},
        every axis/.append style={line width = 2.5pt},
        ]
        \addplot [mark=o,color=c10] table[x=k,y=BCC]{\CASE};
        \legend{ \huge BCC, \huge No. authors, \huge No. collaborators};
        \end{axis}
    \end{tikzpicture}
    }
    \subfigure[authors - collaborators]{
    \begin{tikzpicture}[scale=0.38]
        \begin{axis}[
          legend style = {
            legend columns=-1,
            draw=none,
        },
        width=0.55\textwidth,
        height=0.36\textwidth,
        xtick = {0, 10, 20, 30, 40, 50, 60},
        xticklabels = {1985, 1990, 1995, 2000, 2005, 2010, 2015},
        mark size=2.0pt, 
        ymax = 9,
        ymin = 1,
        ytick = {1, 3, 5, 7},
        yticklabels = {1.0, 3.0, 5.0, 7.0},
        ylabel={\Huge \bf value},
        xlabel={\Huge \bf years}, 
        ticklabel style={font=\Huge},
        every axis plot/.append style={line width = 2.5pt},
        every axis/.append style={line width = 2.5pt},
        ]
        \addplot [mark=o,color=c2] table[x=k,y=auth]{\CASE};
        \addplot [mark=o,color=c8] table[x=k,y=cobo]{\CASE};
        \legend{ \huge No. authors, \huge No. collaborators};
        \end{axis}
    \end{tikzpicture}
    }
    \caption{Case study (1): the trend of closeness in global research collaboration}
    \label{fig:case-study-1}
\end{figure}
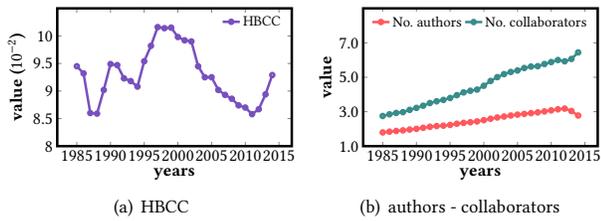

\section{CONCLUSIONS}
\label{sec:conclusion}
This paper introduces the historical butterfly counting problem on temporal bipartite graphs. We propose a graph structure-aware indexing approach to facilitate efficient query processing with a tunable balance between query time and memory cost. Besides, we theoretically prove that our approach is especially advantageous on power-law graphs, by breaking the conventional complexity barrier associated with general graphs. To further mitigate the index overhead, we devised a high-quality approximation algorithm that leverages a compressed index structure, thereby enhancing overall efficiency. Extensive empirical tests show that our algorithm is up to five orders of magnitude faster than existing algorithms with controllable memory costs. 
\clearpage
\bibliographystyle{ACM-Reference-Format}
\balance
\bibliography{References}
\newpage
\clearpage
\newpage
\appendix
\newpage
\section{Omitted Proofs for Index Compressions}
\label{append:compression}
\begin{proof}[Proof of \cref{thm:SGSI}]
    The approximation $B'$ is unbiased because each butterfly from $T_E$ that contributes to $B$ is preserved with probability $\lambda_1$, and each of these preserved butterfly contributes $1/\lambda_1$ to $B'.$
    
    $B'$ is the sum of independent random variables taking values in $\{0, 1/\lambda_1\}.$ By Chernoff bound, we have 
    \begin{align*}
    \pr{}{|B'-\expec{}{B'}|\ge \delta\expec{}{B'}} = O(\exp(-\delta^2\lambda_1\expec{}{B'}/3)).
    \end{align*} Since $B'$ is unbiased, $\expec{}{B'}=B.$ Substituting $\delta=c\log^{0.5}(n)\lambda_1^{-0.5}{B}^{-0.5}$ for some large enough $c$ into the equation above gives
    \begin{align*}
    &\pr{}{|B'-B|\ge \delta B} \\
    =&O(\exp(-\delta^2 \lambda_1 B/3))\\
    =&O(\exp(-c^2\log(n)))=O(n^{-c^2}).
    \end{align*}
\end{proof}

\begin{proof}[Proof of \cref{thm:DGSI}]
    We consider any two different wedges $w_1, w_2$ from the same $\mathcal{CS}\in T_C$ such that both appears in the projected graph $G_{[t_s, t_e]}$ for the time-window $[t_s, t_e].$ Let $\widetilde{\mathcal{CS}}$ be the corresponding data structure of $\mathcal{CS}$ in $\widetilde{T_C}.$ Let $C'$ be the number of wedges in $G_{[t_s, t_e]}$ that are inserted into $\widetilde{\mathcal{CS}}$. Let $C$ be the total number of wedges in $G_{[t_s, t_e]}$ that are inserted into $\mathcal{CS}$. The pair of wedges, $(w_1, w_2)$, contributes $1/\lambda_2^2$ to the answer $\binom{C'}{2}/\lambda_2^2$ if both of them are preserved $\widetilde{\mathcal{CS}}$. The expected contribution is exactly $1$ because $w_1$ and $w_2$ are preserved independently with probability $\lambda_2$ each. Thus, by the linearity of expectation, $\expec{}{\binom{C'}{2}/\lambda_2^2}$ is exactly $\binom{C}{2},$ i.e., the number of pairs $(w_1, w_2).$ Summing over all $\mathcal{CS}\in T_C$ gives the unbiasedness.

    Let $\delta$ be an undetermined coefficient. By Chernoff bound, 
    \begin{align*}
    &\pr{}{|C'-\lambda_2C| \ge \delta C}\\
    \le & O(\exp(-\Omega(-\delta^2\lambda_2C))).
    \end{align*}
    Let $\delta=\log^{0.5}(n)\lambda_2^{-0.5}C^{-0.5}.$ We have that with high probability, $|C'-\lambda_2C|=O(C^{0.5}\log^{0.5}(n)\lambda_2^{-0.5}).$
    This means that $\binom{C'}{2}/\lambda_2^2$ is within multiplicative error $O(C^{-0.5}\lambda_2^{-1.5}\log^{0.5}(n))$ of the correct answer $\binom{C}{2}.$
    
    Conditioning on the error bound above hold for every $\mathcal{CS}$ in $T_C,$ we may use Hoeffding's inequality. For each $\mathcal{CS}_i\in T_C$ such that $|\mathcal{CS}_i|\ge 2$, $\binom{C_i'}{2}/\lambda_2^2$ is an independent random variable taking values from 
    \begin{align*}
    \binom{C_i}{2}\pm O(C^{1.5}\lambda_2^{-3.5}\log^{0.5}(n)),
    \end{align*} where $C_i$ and $C'_i$ are the number of wedges in $\mathcal{CS}_i\cap G_{[t_s, t_e]}$ and the number of preserved wedges in $\mathcal{CS}_i\cap G_{[t_s, t_e]}$. Let $S'=\sum_{i}\binom{C_i'}{2}/\lambda_2^2,$ and let $S=\sum_{i} \binom{C_i}{2}.$ We have, for any parameter $t$, \begin{align*}
        &\pr{}{|S'-S|\ge t}\\
        \le & O\left(\exp\left(-\Omega\left(t^2/\left(\lambda_2^{-3.5}\log^{0.5}(n)\sum_{i}C_i^{3}\right)\right)\right)\right)
    \end{align*} by Hoeffding's inequality. We may set $t=c\lambda_2^{-1.75} \log^{0.75}(n) S^{0.75}$ for large enough constant $c$ to ensure that $|S'-S|< t$ with high probability.
\end{proof}

\section{Handling Duplicate Edges}
\label{append:duplicate}

The key challenge when the graph has duplicate edges is that the life cycle of a wedge or a butterfly is no longer an active timestamp as defined in \Cref{def:original-timestamp}. As we will see in our technical report Part \ref{append:duplicate}, the life cycle can be decomposed into several redefined active timestamps (\Cref{def:active-timestamp}) for graphs with duplicate edges. We will prove that the decomposition does not increase the time complexity or memory usage of \GSI~(\Cref{lemma:duplicate} and \Cref{theo:time-dup}).

\begin{mydef}[Active Intervals for Graphs with Duplicate Edges]
\label{def:active-timestamp}
Given a bipartite temporal graph $G$ with duplicate edges, a subgraph $P$, we define the \textit{active intervals} $\tilde{\mathcal{T}}(P)$ as a tuple of timestamps of the form $[l, r_1, r_2] (l \leq r_1 \leq r_2)$ such that
$P$ is active in the query time-window $[t_s, t_e]$ if and only if for exactly one of the timestamps $[l, r_1, r_2]$, $t_s \leq l$ and $r_1 \leq t_e < r_2$.
\end{mydef}

For applying \texttt{GSI} to a graph with duplicate edges, we first modify $\mathcal{CS}$ such that it can answer 2D-range queries on timestamps of the form $[l, r_1, r_2],$ i.e., counting the number of $[l,r_1,r_2]$ such that $t_s \leq l$ and $r_1 \leq t_e \leq r_2$ given the query time-window $[t_s, t_e].$ This can be done by applying the inclusive-exclusive principle on 2D ranges. Then we generate the active intervals (tuple of timestamps $[l, r_1, r_2]$) for each wedge. We feed each $[l, r_1, r_2]$ to \GSI~(with the modified $\mathcal{CS}$). 

Now, we provide the intuition and detailed analysis of \Cref{alg:active-timestamp}.

Given a butterfly or wedge $P$, we consider a subgraph $S$ constructed by choosing exactly one edge $(u,v,t)$ for each $(u,v) \in E(P)$. We denote the set of all possible $S$s by $\mathcal{S}.$ Due to duplicate edges, $|\mathcal{S}|>1.$ For each $S\in \mathcal{S}$, its active timestamp is $[l_S, r_S]$ where $l_S=\min_{(u, v, t) \in S}t$ and $r_S=\max_{(u, v, t) \in S}t$. For a time-window $[t_s,t_e]$, $P \in G_{[t_s, t_e]}$ if and only if there exists an $S\in \mathcal{S}$ satisfying $t_s \leq l_S \leq r_S \leq t_e$. Therefore, for $S_1, S_2\in \mathcal{S}$, if $l_{S_1} \leq l_{S_2} \leq r_{S_2} \leq r_{S_1}$, then $S_1$ is not necessary to be considered for any query time-window. In this way, we manage to reduce the number of $S$s to be stored to answer historical queries concerning $P$.

We still need to resolve the issue of overcounting. For a time-window $[t_s,t_e]$, if there are multiple $S$s such that their timestamps are all included in $[t_s,t_e]$, we should not count $P$ as multiple wedges or butterflies (\Cref{def:projected-graph}). This is the reason we consider the redefined active timestamp in \Cref{def:active-timestamp}. Let the reduced set of $S$s be $\{[l_1,r_1],[l_2,r_2],\cdots,[l_k,r_k]\}$ satisfying $l_i<l_{i+1}, r_i\le r_{i+1}$ for each $1\le i<k$. We create a redefined timestamp $[l_i, r_i, r_{i+1}]$ for each $1\le i < k$ and a redefined timestamp $[l_k, r_k, \infty]$ for the last timestamp $[l_k, r_k]$. We can see that exactly one of the redefined timestamps becomes active when $P\in G_{[t_s, t_e]}:$ If there exists $1\le i<k$ such that $r_i\le t_e< r_{i+1}$, only $[l_i, r_i, r_{i+1}]$ is active. Otherwise, we have $t_e \ge r_k$. In such case, $[l_k, r_k, \infty]$ is active.
In addition, if $P\not \in G_{[t_s, t_e]},$ no redefined timestamps becomes active. For any timestamp $[l_i, r_i],$ $[t_s, t_e]$ does not include $[l_i, r_i].$ This implies that the redefined timestamp $[l_i, r_i, r_{i+1}]$ (or $[l_k, r_k, \infty]$ when $i=k$) is not active.

\begin{algorithm}[t]
\caption{\small \textsc{Computing $\left \langle x \leadsto y \leadsto z \right \rangle$'s Active Timestamps}}\label{alg:active-timestamp}
\KwIn{A list of unique timestamps on the duplicate edges between $x$ and $y$: $L_{x, y}$;
A list of unique timestamps on the duplicate edges between $y$ and $z$ : $L_{y, z}$;
}
\KwOut{The active intervals $A = \{ [l_i, r_{1,i}, r_{2,i}]\}$ for $\left \langle x \leadsto y \leadsto z \right \rangle$}
$L_{x,y} \leftarrow L_{x,y} \cup \{\infty\}, L_{y,z} \leftarrow L_{y,z} \cup \{\infty\}$; \\ \label{AT-inf}
Sort $L_{x, y}$ and $L_{y, z}$ in ascending order; \\ \label{AT-sort}
$S \leftarrow \emptyset, i \leftarrow 1, j \leftarrow 1$; \\ \label{AT-newS}
\If{$L_{x, y}[1] \leq L_{y,z}[1]$}{
    $i \leftarrow $ the maximum $k$ such that $L_{x, y}[k] \leq L_{y, z}[1]$; \label{AT-movei1}
}
\Else{
    $j \leftarrow $ the maximum $k$ such that $L_{y, z}[k] \leq L_{x, y}[1]$; \label{AT-movej1}
}
$l \leftarrow \infty, r \leftarrow \infty$; \\
\While{$i < |L_{x, y}|$ and $j < |L_{y, z}|$}{ \label{AT-while}
    $l' \leftarrow \min(L_{x,y}[i], L_{y,z}[j]), r' \leftarrow \max(L_{x,y}[i], L_{y,z}[j])$; \\ \label{AT-minmax}
    \lIf{$r \neq \infty$}{
        $A = A \cup \{[l, r, r']\}$ \label{AT-insert-S} 
    }
    \If{$L_{x,y}[i] \leq L_{y, z}[j]$}{ \label{AT-ileqj}
        $i \leftarrow i + 1$; \\ \label{AT-movei}
        $j \leftarrow $ the maximum $k$ such that $L_{y, z}[k] \leq L_{x, y}[i]$; \label{AT-adjustj}
    }
    \Else{
        $j \leftarrow j + 1$; \\ \label{AT-movej}
        $i \leftarrow $ the maximum $k$ such that $L_{x, y}[k] \leq L_{y, z}[j]$; \label{AT-adjusti}
    }
    $l \leftarrow l', r \leftarrow r'$; \label{AT-new-lr}
}
\If{$r \neq \infty$}{
    $A = A \cup \{[l, r, \infty]\}$
}

\Return $S$; \label{AT-return}
\end{algorithm}

\Cref{alg:active-timestamp} computes the active intervals for a wedge $\left \langle x \leadsto y \leadsto z \right \rangle$. Given the two timestamp sets $L_{x,y}$ and $L_{y,z}$, we will produce all necessary $S\in \mathcal{S}$ and store the redefined timestamps in $A$. 
We can safely assume that there is no duplicate timestamp in either $L_{x,y}$ or $L_{y,z}$. To begin with, we initialize $S$ to be $\emptyset$ (\Cref{AT-newS}). We sort $L_{x,y}$ and $L_{y,z}$ from small to large (\Cref{AT-sort}). 
We also insert an $\infty$ timestamp to both timestamp sets (\Cref{AT-inf}) to avoid boundary cases. 
We will enumerate every element in these sets, and we denote them as $L_{x,y}[i] (1 \leq i \leq |L_{x,y}|$ and $L_{y,z}[j] (1 \leq j \leq |L_{y,z}|$. Consider a pair $(i,j)$. 
When $L_{x,y}[i] \leq L_{y,z}[j]$, we will adjust $i$ to $i'$ such that $L_{x,y}[i+1] > L_{y,z}[j]$. That is to say, we will adjust the smaller side as much as possible without breaking the inequality $L_{x,y}[i] \leq L_{y,z}[j]$ (\Cref{AT-movei1}, \Cref{AT-adjustj}, \Cref{AT-movej1}, and \Cref{AT-adjusti}). 
We do the same for the case which $L_{x,y}[i] > L_{y,z}[j]$ (adjusting $j$ instead). 
For the current enumerated pair $(i,j)$, the corresponding redefined timestamp of the wedge formed by $(x,y,L_{x,y}[i])$ and $(y,z,L_{y,z}[j])$ is $[l, r, \max(L_{x,y}[i], L_{y,z}[j])]$, where $l$ and $r$ are the minima and maxima of the two timestamps $L_{x,y}[i]$s and $L_{y,z}[j]$s in the previous iteration (\Cref{AT-minmax} and \Cref{AT-insert-S}). 
We insert it into $A$ (\Cref{AT-insert-S}). 
After inserting, we move to the next pair by increasing one of the indices $i$ or $j$. 
When $L_{x,y}[i] \leq L_{y,z}[j]$, by setting $i$ to $i+1$, we will have $L_{x,y}[i] > L_{y,z}[j]$ (\Cref{AT-ileqj} to \Cref{AT-adjustj}). 
Then we adjust $j$ again to the largest possible value satisfying $L_{y,z}[j] \le L_{x,y}[i]$ (\Cref{AT-adjustj}). 
If $L_{x,y}[i] > L_{y,z}[j]$, we increase $j$ and adjust $i$ instead (\Cref{AT-movei} and \Cref{AT-adjusti}). 
The whole process will terminate when no more feasible pair needs to be enumerated (\Cref{AT-while}). Then, all the redefined timestamps are stored in $A$. We return $A$ as the output (\Cref{AT-return}).

Lastly, we prove that counting the number of activated redefined timestamps can be converted to computing the difference between two 2D-range counting queries.

Previously, we regarded the active timestamp as a single point in the 2-D plane and inserted it into the 2D-counting data structure. For any time window, the answer is computed by a single query on the data structure. Under the current definition of the active timestamp, we can still efficiently answer the query using two 2D-range counting data structures $T$ and $\overline{T}$ instead. That is to say, after construction, we will be able to answer the number of redefined timestamps $[l, r_1, r_2]$ satisfying $t_s \leq l$ and $r_1 \leq t_e < r_2$. Specifically, for $[l,r_1,r_2]$, we insert $(l,r_1)$ and $(l, r_2)$ into $T$ and $\overline{T}$ respectively. To count the number of active timestamps of a time-window $[t_s,t_e]$, we first query $[t_s, \infty] \times [-\infty, t_e]$ on both $T$ and $\overline{T}$, denoted as $num$ and $\overline{num}$ respectively. Then, we return $num - \overline{num}$ as the answer. The formal proof is as follows.
\begin{align*}
&\sum_{P \in G}\  \sum_{[l_i, r_{1,i}, r_{2,i}] \in \tilde{\mathcal{T}}(P)} \mathbbm{1}\{ l_i \geq t_s \land t_e \in [r_{1,i}, r_{2,i})\} = \\
&\sum_{P \in G}\  \sum_{[l_i, r_{1,i}, r_{2,i}] \in \tilde{\mathcal{T}}(P)} \mathbbm{1}\{ l_i \geq t_s \land t_e \geq r_{1,i}\} - \mathbbm{1}\{ l_i \geq t_s \land t_e \geq r_{2,i} \} =\\
&\left ( \sum_{P \in G}\  \sum_{[l_i, r_{1,i}] \in \tilde{\mathcal{T}}(P)} \mathbbm{1}\{ l_i \geq t_s \land t_e \geq r_{1,i}\} \right )- \\ 
&\left ( \sum_{P \in G}\  \sum_{[l_i, r_{2,i}] \in \tilde{\mathcal{T}}(P)} \mathbbm{1}\{ l_i \geq t_s \land t_e \geq r_{2,i}\} \right )= \\
&T.query([t_s, \infty] \times [-\infty, t_e]) - \overline{T}.query([t_s, \infty] \times [-\infty, t_e])
\end{align*}

Lastly, we prove that if the size (number of $[l, r_1, r_2]$ in the tuple) of the active intervals of each wedge is bounded (\Cref{lemma:duplicate}), both \texttt{EBI} and \texttt{CBI}'s time complexity will not be compromised.

\begin{lemma} \label{lemma:duplicate}
For any two vertices $u,v \in V(G)$, we denote $cnt_{u,v}$ as the number of edges $(u,v,t) \in E(G)$.
There exists an algorithm (\Cref{alg:active-timestamp} in the technical report Part \ref{append:duplicate}) that returns its active intervals (\Cref{def:active-timestamp}) of size $O(\min(cnt_{x,y}, cnt_{y,z}))$ for any wedge $\left \langle x \leadsto y \leadsto z \right \rangle$.
\end{lemma}

\begin{proof}[Proof of \Cref{lemma:duplicate}]
WLOG, we assume $cnt_{x, y} \leq cnt_{y, z}$. 
It is easy to see the time complexity of \Cref{alg:active-timestamp} is $O(|S|\log m)$, and we will then show $|S| \leq 3cnt_{x, y}$.

For each wedge $\left \langle x \leadsto y \leadsto z \right \rangle [l_i, r_{1,i}, r_{2,i}] \in S$, we have at least one of  $l_i$ or $r_{1,i}$ comes from a timestamp of $L_{x, y}$.
In other words, for each pair $(l, r)$ in lines~\ref{AT-while} to \ref{AT-new-lr} of \Cref{alg:active-timestamp}, there exists an $i \in [1, cnt_{x,y}]$ such that at least one of $l$ or $r$ equals to $L_{x, y}[i]$. We consider these cases as follows:
\begin{itemize}[leftmargin=*]
    \item $l = L_{x, y}[i], r \neq L_{x, y}[i]$: We move $i \leftarrow i + 1$ in the next iteration, which implies this case will only occur for at most $cnt_{x, y}$ times.
    \item $l \neq L_{x, y}[i], r = L_{x, y}[i]$: let $\hat{l}$ and $\hat{r}$ be the values of $l$ and $r$ in the next iteration, respectively. We have $\hat{r} \geq \hat{l} \geq r$, which implies this case will only occur for at most $cnt_{x, y}$ times.
    \item $l = L_{x, y}[i], r = L_{x, y}[i]$: let $\hat{l}$ and $\hat{r}$ be the value of $l$ and $r$ in the next iteration, respectively. We have either $r' \geq \hat{l} > r$ or $\hat{r} > \hat{l} \geq r$, which implies this case will only occur for at most $cnt_{x, y}$ times.
\end{itemize}
Since each wedge belongs to one of the three cases above, we prove that $|S| \leq 3cnt_{x,y}$.

\end{proof}

With this lemma, we are ready to prove that the time complexity for our algorithms will not be compromised by duplicate edges:

\begin{theorem}[Time complexity with Duplicate Edges]
\label{theo:time-dup}
\textbf{(i)} There exists a modification for \texttt{EBI} that can run in $O(\log m)$ time and $O(m^2)$ memory usage for bipartite temporal graphs with duplicate edges; \textbf{(ii)} There exists a modification for \texttt{CBI} that can run in $O(\tilde{w}\log m)$ time and $O(m\delta)$ memory usage for bipartite temporal graphs with duplicate edges.
\end{theorem}

\begin{proof}[Proof of \Cref{theo:time-dup}]
\textbf{(i)}: Considering a butterfly consisting of two wedges $\left \langle x \leadsto y_1 \leadsto z \right \rangle$ and $\left \langle x \leadsto y_2 \leadsto z \right \rangle$, the size of its active intervals is $O(\min(cnt_{x,y_1}, cnt_{y_1, z}) \times \min(cnt_{x, y_2}, cnt_{y_2, z}))$ by \Cref{lemma:duplicate}. Since
\begin{align*}
&\sum_{\text{each butterfly } \left \langle x, y_1, z, y_2 \right \rangle \in G} \min(cnt_{x,y_1}, cnt_{y_1, z}) \times \min(cnt_{x, y_2}, cnt_{y_2, z}) \\ &\leq \sum_{\text{each butterfly } \left \langle x, y_1, z, y_2 \right \rangle \in G} cnt_{x,y_1}cnt_{y_2, z} \leq m^2.    
\end{align*}
, the bound of the number of points in $\mathcal{CS}s$ $O(m^2)$ for \texttt{EBI} still holds.

\textbf{(ii)}: Considering a wedge $\left \langle x \leadsto y \leadsto z \right \rangle$, the the size of its active intervals in its active timestamp is $O(\min(cnt_{x,y}, cnt_{y, z}))$ by \Cref{lemma:duplicate}.
The total number of tuples for all wedges' active timestamps is:
\begin{align*}
    &\sum_{y \in V(G)}\sum_{x \in N(y), \atop pr(x) \prec pr(y)}\sum_{z \in N(y), \atop pr(x) \prec pr(z)} \min(cnt_{x, y}, cnt_{y, z}) \\
    &\leq \sum_{y \in V(G)}\sum_{x \in N(y), \atop pr(x) \prec pr(y)} cnt_{x, y} deg_x \\
    &\leq \sum_{(x, y) \in E(G)} \min(deg_x, deg_y) = m\delta.
\end{align*}
Therefore, the bound of the number of points in $\mathcal{CS}s$ for \texttt{CBI}, $O(m\delta)$, also holds.
\end{proof}

\section{Omitted Proof for Power-law Graphs}
\label{sec:append-power-law}
\begin{proof} [Proof for \Cref{thm:double_power_law} and \Cref{thm:single_power_law}]
By \Cref{thm:complexity}, we know that the query time for \GSI~is nearly linear in the number of keys $(x, z)$ in $T_C.$
The space usage for \GSI is bounded by the number of butterflies not maintained by $T_C[(x, z)]$, plus the total number of wedges in each $W[(x, z)]$ for $(x, z)\in T_C.keys().$

To bound the query time and space usage, let $k$ be a parameter between $1$ and $\Delta=\max(\Delta_1, \Delta_2)$. Let $P_{\ge k}$ be the set of unordered pairs $(x, z)$ such that $x, z$ are on the same side of the bipartite graph $G$, and that $d_x, d_z \ge k$. Let $\#\Join_{\ge k}$ be the number of butterflies $B$ such that $\exists \{x, z\} \in P_{\ge k}$, $\{x, z\} \subseteq B$. Here we abuse notation and use $B$ to mean the vertices of a butterfly $B$. 
 
 In \Cref{gsi-c:enumeratewkey}~of \Cref{alg:GSI-Construction}, we construct $T_C[(x, z)]$s for pairs $(x, z)$ with the largest $W[(x, z)]$s until the total number of butterflies in these largest $W[(x, z)]$s exceeds $(1-\alpha)numB$. The rest of the butterflies will be maintained by $T_E.$

Intuitively, for proving the efficiency of \GSI, we would like to show that a small number of $\{x,z\}$ (those in $P_{\ge k}$) covers a large fraction ($\#\Join_{\ge k}/numB$) of the total number of butterflies. Formally, we can prove that for some $k$, \GSI~has query complexity $\Otil(|P_{\ge k}|)$ and expected space complexity.

$O\left(\#\Join_{\ge 1}-\#\Join_{\ge k}+\sum_{(x, z)\in P_{\ge k}}|T_C[(x, z)]|\right)$. To show this, let's consider an algorithm similar to \Cref{alg:GSI-Construction}. The modified algorithm changes the condition on \Cref{gsi-c:ifbig} from ``$num \geq \alpha \cdot numB$'' to ``$\{x, z\} \in P_{\ge k}$''. Then $T_C$ contains $|P_{\ge k}|$ elements (one for each $(x, z)\in P_{\ge k}$), and $T_E$ contains $\#\Join_{\ge 1}-\#\Join_{\ge k}$ butterflies. The time complexity of the modified algorithm is $\Otil(|P_{\ge k}|)$ and that the expected space usage of it is

$O\left(\#\Join_{\ge 1}-\#\Join_{\ge k}+\sum_{(x, z)\in P_{\ge k}}|T_C[(x, z)]|\right)$.

In the original \GSI~(\Cref{alg:GSI-Construction}), for a fixed choice of $k$, we can choose $\alpha$ properly such that the condition on \Cref{gsi-c:ifbig} evaluates to true for the first $|P_{\ge k}|$ iterations, i.e., after the first $|P_{\ge k}|$ iterations of the loop, $num$ is no more than $\alpha\cdot numB$. The query time of \GSI~is bounded by $\Otil(|P_{\ge k}|).$ The space usage of \GSI~is bounded above by that of the modified algorithm, because the sets in $W$ are sorted with decreasing order of sizes. Each set $W[(x, z)]$ costs $|W[(x, z)]|$ space if it is maintained in $T_C$ and $\binom{|W[(x, z)]|}{2}$ space if it is maintained in $T_E$. \GSI~costs less space because it maintains larger sets in $T_C$, compared to the modified algorithm. These will automatically translates to the same bounds for \GSI.

For simplicity, we use $V_1, V_2$ to denote $U$, $L$.
We first calculate the expected number of edges, $m$, of $G$. $m$ can be calculated by $\Delta_i,\gamma_i, n_i$ for either $i=1$ or $i=2$. For the model to be consistent, we require that the $m$'s calculated by $i=1$ and $i=2$ are equal.

\begin{lemma}
     For any $i\in \{1,2\}$, if $\gamma_i\in (2, 3)$, 
     $$m=\left(1+O\left(\Delta_i^{2-\gamma_i}\right)\right)\frac{n_i}{s_i(\gamma_i-2)}.$$
\end{lemma}
\begin{proof}
    Let $x$ be a fixed vertex in $V_i$.
    \begin{align*}
        m=&n_i\expec{}{\deg_x} \tag{by linearity of expectation}\\
        =&n_i\expec{}{d_x}\\
        =&n_i\frac{1}{s_i}\sum_{i=1}^{\Delta_i} i^{1-\gamma_i}\\
        =&\left(1+O\left(\Delta_i^{2-\gamma_i}\right)\right) \frac{n_i}{s_i(\gamma_i-2)}.
    \end{align*}
\end{proof}

Next, we estimate the expectation of $|P_{\ge k}|$ for bounding the query time.
\begin{lemma}
\label{lem:expec_p_ge_k}
    $\expec{}{|P_{\ge k}|}=\left(1+O\left(\frac{1}{n_1}+\frac{1}{n_2}\right)\right)\left(\frac{n_1^2s_{1, k}^2}{2s_1^2}+\frac{n_2^2s_{2, k}^2}{2s_2^2}\right).$
\end{lemma}
\begin{proof}
    Recall that for any pair of vertices $x, z$, $(x, z)\in P_{\ge k}$ if $x, z\in V_i$ and $d_x, d_z\ge k$. The expected number of such pairs for a fixed $i\in \{1, 2\}$ is 
    \begin{align*}
        &\binom{n_i}{2}\pr{}{d_x\ge k}\pr{}{d_z\ge k}\\
        =&\binom{n_i}{2}\left(\frac{\sum_{d=k}^{\Delta_d}d^{-\gamma_i}}{s_i}\right)^2\\
        =&\left(1+O\left(\frac{1}{n_i}\right)\right)\frac{n_i^2s_{i, k}^2}{2s_i^2}.
    \end{align*}
    Summing over $i=\{1, 2\}$ gives
    $$\expec{}{|P_{\ge k}|}=\left(1+O\left(\frac{1}{n_1}+\frac{1}{n_2}\right)\right)\left(\frac{n_1^2s_{1, k}^2}{2s_1^2}+\frac{n_2^2s_{2, k}^2}{2s_2^2}\right).$$
\end{proof}

Lastly, we calculate $\expec{}{\#\Join_{\ge k}},$ and $\expec{}{\#\Join_{\ge 1}}.$ The difference of these two numbers will bound the space usage.
\begin{lemma}
\label{lem:butterfly_ge_k}
    \begin{align*}
        &\expec{}{\#\Join_{\ge k}}\\
    \ge &\left(1 + O\left(\Delta_1^{\gamma_1 - 3}+\Delta_2^{\gamma_2 - 3}\right)\right)\frac{n_1^2n_2^2}{4m^4}\\& \frac{\left(\Delta_1^{3-\gamma_1}-k^{3-\gamma_1}\right)\left(\Delta_2^{3-\gamma_2}-k^{3-\gamma_2}\right)}{s_1s_2(3-\gamma_1)(3-\gamma_2)}.
     \end{align*}
\end{lemma}
\begin{proof}
    To get a lower bound for $\expec{}{\#\Join_{\ge k}}$, we use the quantity $\#\Join_{both \ge k}$ be the number of butterflies $B=(x, y, z, w)$ such that both $(x, z)$ and $(y, w)$ are from $P_{\ge k}.$
    
    $\#\Join_{both \ge k}$ is no more than $\#\Join_{\ge k}$ which is the number of butterflies such that at least one of $(x, z)$ and $(y, w)$ are from $P_{\ge k}.$
    
    We sum up the probability that $x, y, z, w$ form a butterfly for $x, z\in V_1, y, w\in V_2$.
    \begin{align*}
        &\expec{}{\#\Join_{both \ge k}}\\
        =&\binom{n_1}{2}\binom{n_2}{2}\sum_{d_1, d_3 \in [k, \Delta_1], d_2, d_4\in [k, \Delta_2]} \pr{}{d_x=d_1}\pr{}{d_y=d_2}\\
        &\pr{}{d_z=d_3}\pr{}{d_w=d_4}\frac{d_xd_y}{m}\frac{d_xd_w}{m}\frac{d_zd_y}{m}\frac{d_zd_w}{m}\\
        =&\binom{n_1}{2}\binom{n_2}{2}\sum_{d_1, d_3 \in [k, \Delta_1], d_2, d_4\in [k, \Delta_2]} \pr{}{d_x=d_1}\pr{}{d_y=d_2}\\
        &\pr{}{d_z=d_3}\pr{}{d_w=d_4}\frac{d_x^2d_y^2d_z^2d_w^2}{m^4}\\
        =&\binom{n_1}{2}\binom{n_2}{2}\frac{1}{m^4} \left(\sum_{d_1, d_3 \in [k, \Delta_1]} \pr{}{d_x=d_1}\pr{}{d_z=d_3} d_x^2d_z^2\right)\\
        &\left(\sum_{d_2, d_4 \in [k, \Delta_2]} \pr{}{d_y=d_2}\pr{}{d_w=d_4} d_y^2d_w^2\right)\\
        =&\binom{n_1}{2}\binom{n_2}{2}\frac{1}{m^4} \left(\sum_{d_1\in [k, \Delta_1]} \pr{}{d_x=d_1}d_x^2\right)^2\\
        &\left(\sum_{d_2\in [k, \Delta_2]} \pr{}{d_y=d_2}d_y^2\right)^2\\
        =&\binom{n_1}{2}\binom{n_2}{2}\frac{1}{m^4} \mathbb{E}^2\left[\mathbf{1}_{d_{v_1}\ge k}d_{v_1}^2\right] \mathbb{E}^2\left[\mathbf{1}_{d_{v_2}\ge k}d_{v_2}^2\right]
    \end{align*} where $v_1$ ($v_2$) is an arbitrary vertex from $V_1$ ($V_2$).
    We next calculate the expectations in the equation above. For any $i\in \{1, 2\}$, 
    \begin{align*}
        &\mathbb{E}\left[\mathbf{1}_{d_{v_i}\ge k}d_{v_i}^2\right]\\
        =&\sum_{d=k}^{\Delta_i} \pr{}{d_{v_i}=d} d^2\\
        =&\frac{1}{s_{i}} \sum_{d=k}^\Delta d^{2-\gamma_i}\\
        =&\left(1 + O\left(\Delta_i^{\gamma_i - 3}\right)\right)\frac{\Delta_i^{3-\gamma_i}-k^{3-\gamma_i}}{s_i(3-\gamma_i)}.
    \end{align*}
    Thus, we have
    \begin{align*}
        &\expec{}{\#\Join_{both \ge k}}\\
        =&\binom{n_1}{2}\binom{n_2}{2}\frac{1}{m^4} \mathbb{E}^2\left[\mathbf{1}_{d_{v_1}\ge k}d_{v_1}^2\right] \mathbb{E}^2\left[\mathbf{1}_{d_{v_2}\ge k}d_{v_2}^2\right]\\
        =&\binom{n_1}{2}\binom{n_2}{2}\frac{1}{m^4} \Pi_{i=1}^2\left(1 + O\left(\Delta_i^{\gamma_i - 3}\right)\right)\frac{\Delta_i^{3-\gamma_i}-k^{3-\gamma_i}}{s_i(3-\gamma_i)}\\
        =&\left(1 + O\left(\Delta_1^{\gamma_1 - 3}+\Delta_2^{\gamma_2 - 3}\right)\right)\frac{n_1^2n_2^2}{4m^4}\\& \frac{\left(\Delta_1^{3-\gamma_1}-k^{3-\gamma_1}\right)\left(\Delta_2^{3-\gamma_2}-k^{3-\gamma_2}\right)}{s_1s_2(3-\gamma_1)(3-\gamma_2)}.
    \end{align*}
\end{proof}

Similar to \Cref{lem:butterfly_ge_k}, we can estimate the expectation of $numB=\#\Join_{\ge k}.$

\begin{lemma}
\label{lem:butterfly_ge_1}
    \begin{align*}
        &\expec{}{\#\Join_{\ge 1}}\\
    =&\left(1 + O\left(\Delta_1^{\gamma_1 - 3}+\Delta_2^{\gamma_2 - 3}\right)\right)\frac{n_1^2n_2^2}{4m^4}\\& \frac{\Delta_1^{3-\gamma_1}\Delta_2^{3-\gamma_2}}{s_1s_2(3-\gamma_1)(3-\gamma_2)}.
     \end{align*}
\end{lemma}
\begin{proof}
The proof is identical to that of \Cref{lem:butterfly_ge_k} with $k=1.$ We may replace the $\ge$ to $=$ because when $k=1$, $\#\Join_{both \ge k}$ is equal to $\#\Join_{\ge k}.$
\end{proof}

\paragraph{Double-sided power-law bipartite graphs}

In the double-sided power-law model, both $\gamma_1$ and $\gamma_2$ are in the range $(2, 3)$. In this case, we have that $s_i$ is a constant for $i=1,2$.
We also have that $m=O(n).$

We define $n=max(n_1, n_2)$, $\Delta=max(\Delta_1, \Delta_2)$, and $\gamma=min(\gamma_1, \gamma_2).$

We may calculate 

$\expec{}{\#\Join_{\ge 1}-\#\Join_{\ge k}}$ by \Cref{lem:butterfly_ge_1} and \Cref{lem:butterfly_ge_k}.
    \begin{align*}
        &\expec{}{\#\Join_{\ge 1}-\#\Join_{\ge k}}\\
        =&O\left(\left(\Delta_1^{\gamma_1-3}+\Delta_2^{\gamma_2-3}\right)\frac{n_1^2n_2^2}{m^4}\left(k^{3-\gamma_1}\Delta_2^{3-\gamma_2}+k^{3-\gamma_2}\Delta_1^{3-\gamma_1}\right)\right)\\
        =&O\left(\frac{n_1^2n_2^2}{m^4}\left(\left(\frac{k}{\Delta_1}\right)^{3-\gamma_1}\Delta_2^{3-\gamma_2} + \left(\frac{k}{\Delta_2}\right)^{3-\gamma_2}\Delta_1^{3-\gamma_1}+k^{3-\gamma_1}+k^{3-\gamma_2}\right)\right)\\
        =&O\left(\Delta^{6-2\gamma}\right)
        \end{align*} for $k\le \Delta.$
        When $k$ is no more than the smaller of $\Delta_1$ and $\Delta_2$, we have a sharper bound.

    \begin{align*}
        &\expec{}{\#\Join_{\ge 1}-\#\Join_{\ge k}}\\
        =&O\left(\frac{n_1^2n_2^2}{m^4}\left(\Delta_2^{3-\gamma_2} + \Delta_1^{3-\gamma_1}+k^{3-\gamma_1}+k^{3-\gamma_2}\right)\right) \tag{$k\le \Delta_1, k\le \Delta_2$}\\
        =&O\left(\frac{n^4}{m^4} \Delta^{3-\gamma}\right)\\
        =&O\left(\Delta^{3-\gamma}\right)
    \end{align*} for $k \le \min(\Delta_1, \Delta_2)$.
By \Cref{lem:expec_p_ge_k}, 
\begin{align*}
        &\expec{}{|P_{\ge k}|}\\
    =&O\left(n^2k^{2-2\gamma}\right).
\end{align*}
Note that the two bounds above on space usage do not depend on $k$. Thus, we may choose the largest possible $k$ to reduce the query time in each case. 
\begin{itemize}[leftmargin=*]
    \item We may choose $k=\Delta$ so that all butterflies are maintained by $T_E$. The query time is $\Otil(1)$ and the expected space usage is $\expec{}{\#\Join_{\ge 1}}=O(\Delta^{6-2\gamma}).$ \item We may also choose $k=\min(\Delta_1, \Delta_2)$ so that the expected query time is $\Otil(\expec{}{|P_{\ge k}|})=\Otil(n^2\min(\Delta_1, \Delta_2)^{2-2\gamma})$ and the expected space usage for $T_E$ is $O(\Delta^{3-\gamma}).$ Note that we need to consider the space usage of $T_C.$ This can be bounded by  
\begin{align*}
    &\sum_{i=1}^2 n_i \expec{}{\deg_{x_i}^2}\\
    =&\sum_{i=1}^2 n_i \frac{\sum_{d=1}^{\Delta_i}d^{2-\gamma_i}}{s_i}\\
    =&O\left(\sum_{i=1}^2 n_i \Delta_i^{3-\gamma_i}\right)\\
    =&O\left(n\Delta^{3-\gamma}\right)
\end{align*} where $x_i\in V_i$ for $i=1,2.$ The total expected space usage is $O\left(n\Delta^{3-\gamma}\right).$

\end{itemize}
\paragraph{Single-sided power-law bipartite graphs}

In the single-sided power-law model, we have $\gamma_1\in (2, 3)$, $\gamma_2=0$, and $\Delta_1>\Delta_2.$
In this case, $s_1$ is a constant and $s_2=\Delta_2$.
We also have that $m=\Theta(n_1)=\Theta(n_2\Delta_2).$

We define $n=max(n_1, n_2)$, $\Delta=max(\Delta_1, \Delta_2)$.
We choose $ \Delta_2 < k\le \Delta_1$.
We may calculate $\expec{}{\#\Join_{\ge 1}-\#\Join_{\ge k}}$ by \Cref{lem:butterfly_ge_1} and \Cref{lem:butterfly_ge_k}.
    \begin{align}
        &\expec{}{\#\Join_{\ge 1}-\#\Join_{\ge k}}\\
        =&O\left(\left(\Delta_1^{\gamma_1-3}+\Delta_2^{\gamma_2-3}\right)\frac{n_1^2n_2^2}{m^4}\left(k^{3-\gamma_1}\Delta_2^{3-\gamma_2}+k^{3-\gamma_2}\Delta_1^{3-\gamma_1}\right)\right)\\
        =&O\left(\frac{n_1^2n_2^2}{m^4}\left(\left(\frac{k}{\Delta_1}\right)^{3-\gamma_1}\Delta_2^{3-\gamma_2} + \left(\frac{k}{\Delta_2}\right)^{3-\gamma_2}\Delta_1^{3-\gamma_1}+k^{3-\gamma_1}+k^{3-\gamma_2}\right)\right)\\
        =&O\left(\frac{n_2^2}{m^2} \left(\frac{\Delta_1^{6-\gamma_1}}{\Delta_2^3}+\Delta_2^3\right)\right) \tag{$k\le \Delta_1, \Delta_1 > \Delta_2$}\\
        =&O\left(\Delta_2+\left(\frac{\Delta_1^{6-\gamma_1}}{\Delta_2^5}\right)\right). \tag{$\frac{n_2}{m}=\Theta\left(\frac{1}{\Delta_2}\right)$}
    \end{align}
Note that the bound above does not depend on $k$. Thus, it is an upper bound of $\expec{}{\#\Join_{\ge 1}}$. We may set $\alpha>1$ so that \GSI~maintain all butterflies in $T_E$. This results in an expected query time of $\Otil(1)$.
\end{proof}

\balance

\end{document}